\documentclass[12pt]{article} 

\usepackage[utf8]{inputenc} 
\usepackage{amsmath,amsthm,verbatim,amssymb,amsfonts,amscd,graphicx}
\usepackage[colorlinks,citecolor=blue,urlcolor=blue]{hyperref}

\usepackage{geometry} % to change the page dimensions
\usepackage{graphicx} % support the \includegraphics command and options

\usepackage{dsfont}
\usepackage{booktabs} % for much better looking tables
\usepackage{array} % for better arrays (eg matrices) in maths
\usepackage{verbatim} % adds environment for commenting out blocks of text & for better verbatim
\usepackage{subfig} % make it possible to include more than one captioned figure/table in a single float
\usepackage{amsmath}
\usepackage[utf8]{inputenc}
\usepackage[english]{babel}
\usepackage{algorithm}
\usepackage{algorithmic}
\usepackage{bm}

\usepackage[symbol]{footmisc}

\usepackage{fancyhdr} % This should be set AFTER setting up the page geometry
\usepackage{sectsty}
\allsectionsfont{\sffamily\mdseries\upshape} % (See the fntguide.pdf for font help)
\usepackage[nottoc,notlof,notlot]{tocbibind} % Put the bibliography in the ToC
\usepackage{natbib}
\usepackage{bbm}
\usepackage[titles,subfigure]{tocloft} % Alter the style of the Table of Contents
\usepackage{comment}
\usepackage{stmaryrd} % for alternative brackets \llbracket, \rrbracket
\usepackage{enumitem} % for continuing list counter

\newcommand{\p}{{\rm I}\kern-0.18em{\rm P}}
\newcommand{\1}{{\rm 1}\kern-0.24em{\rm I}}
\newcommand{\E}{{\rm I}\kern-0.18em{\rm E}}
\newcommand{\R}{{\rm I}\kern-0.18em{\rm R}}

\newtheorem{theorem}{Theorem}[section]
\newtheorem{lemma}{Lemma}[section]

\newtheorem{proposition}{Proposition}[section]
\newtheorem{ex}{Example}[section]

\newtheorem{definition}{Definition}[section]

\newcommand{\qmq}[1]{\quad\mbox{#1}\quad}

\newcommand{\qm}[1]{\quad\mbox{#1}}
\newcommand{\mA}{\mathcal{A}}
\newcommand{\mC}{\mathcal{C}}
\newcommand{\mM}{\mathcal{M}}
\newcommand{\mS}{\mathcal{S}}
\newcommand{\fC}{\mathfrak{C}}
\newcommand{\AMO}{$\alpha$ max optimal}
\newcommand{\AO}{$\alpha$ optimal}
\newcommand{\To}{\rightarrow}
 \newcommand{\wtilde}[1]{\widetilde{#1}}
 \newcommand{\oa}{\overline{a}}
  \newcommand{\ub}{\underline{b}}

\title{Optimal and fast confidence intervals for hypergeometric successes}
\author{\textsc{Jay Bartroff}$^*$, \textsc{Gary Lorden}$^\dag$, and \textsc{Lijia Wang}$^\ddag$\\
\small{$^*$Department of Statistics and Data Sciences, University of Texas at Austin, Austin, Texas, USA}\\
\small{$^\dag$Department of Mathematics, Caltech, Pasadena, California, USA}\\
\small{$^\ddag$Department of Mathematics, University of Southern California, Los Angeles, California, USA}
}

\begin{document}
\maketitle

\abstract{We present an efficient method of calculating exact confidence intervals for the hypergeometric parameter representing the number of ``successes,'' or ``special items,'' in the population.  The method inverts minimum-width acceptance intervals after shifting them to make their endpoints nondecreasing while preserving their level.  The resulting set of confidence intervals achieves minimum possible average size, and even in comparison with confidence sets not required to be intervals it attains the minimum possible cardinality most of the time, and always within $1$.  The method compares favorably with existing methods not only in the size of the intervals but also in the time required to compute them.  The available \textsf{R} package \texttt{hyperMCI} implements the proposed method.}

\section{Introduction}
Given integers $0<n\le N$ and $0\le M\le N$, a random variable $X$ has the \textit{hypergeometric distribution~$\mbox{Hyper}(M,n,N)$} if
\begin{equation}\label{hyp.pdf}
P_M(X=x) = \left. {M\choose x}{N-M\choose n-x}\right/ {N\choose n}
\end{equation}
for all integer values of $x$ such that the quotient~\eqref{hyp.pdf} is defined, with $P_M(X=x)=0$ otherwise. It is not hard to see that \eqref{hyp.pdf} is nonzero if and only if
\begin{equation}\label{x.range}
x_{\min}:=\max\{0, M+n-N\}\le x\le \min\{M, n\}=:x_{\max}.
\end{equation}

The most common setting in which the hypergeometric distribution arises is when $X$ counts the number of items with a certain binary ``special'' property, sometimes called a ``success,'' in a simple random sample (i.e., sampled uniformly without replacement) of size~$n$ from a population of size~$N$ containing $M$ special items. But the hypergeometric arises in many other ways\footnote{For readers interested in aspects of the hypergeometric distribution not considered here, we refer them to \citet{Hald90} for its history and naming, \citet{Keilson71} for log-concavity and other properties, and \citet{Chvatal79} and \citet{Skala13} for exponential tail bounds,  to name a few.} not involving a simple random sample, such as the analysis of a $2\times 2$ contingency table using Fisher's Exact Test, and in other sampling schemes.  

\subsection{Summary of our approach}
Our approach to constructing $(1 - \alpha)$-confidence intervals for $M$ based on $X$ is by inverting tests of the hypotheses $H: M=M_0$, which we denote as $H(M_0)$, for $M_0=0,1,\ldots, N$. For testing $H(M)$, we utilize  acceptance intervals~$[a_M,b_M]$ that \emph{maximize} the acceptance probability $P_M( X\in [a_M,b_M])$ among all shortest possible level-$\alpha$ intervals, a property we call \textit{\AMO{}} which is discussed in Section~\ref{sec: adj}, along with a novel method of shifting a set of \AMO{} intervals so their endpoints $a_M, b_M$ form nondecreasing sequences. This guarantees that the confidence sets that result from inversion are intervals, which is our goal here. After obtaining and shifting a set of \AMO{} intervals, in Section~\ref{sec:mod.minacc} we discuss how to further modify them to make them symmetrical, and discuss the case $M=N/2$ when $N$ is even, which needs separate handling. With our symmetrical and monotonic acceptance intervals in hand, in Section~\ref{sec:length.CI} we prove the size-optimality results for the confidence intervals that result from inversion. In Section~\ref{sec:ex.comp} we present some numerical examples and compare with two existing methods, including the notable, recent method of W.~Wang~\citeyearpar{Wang15}. There we also apply our method to some data about the air quality in China.

\subsection{Related previous work}

For exact confidence sets, there is much more literature on the related problem of the Binomial success probability than for the hypergeometric, beginning with \citet{Clopper34} who applied the method of pivoting the c.d.f.\ to the Binomial problem. Sterne's~\citeyearpar{Sterne54} method for the Binomial inverts hypothesis tests with the $p$-value as the test statistic, and he observed that the resulting intervals are ``sometimes narrower'' \citep[][p.~278]{Sterne54} than the Clopper-Pearson intervals. Sterne's method can alternatively be described as inverting acceptance intervals with maximal acceptance probability, which is similar to the method we apply here to the hypergeometric.  \citet{Crow56} showed that Sterne's~\citeyearpar{Sterne54} method yields intervals with minimal total (or average) width, but also pointed out some ``irregularities'' in the method, such as occcasionally producing non-intervals, or giving longer intervals for \emph{lower} confidence levels based on the same data. \citet{Crow56} proposed a modification of Sterne's method eliminating these irregularities while maintaining minimal total width. \citet{Blyth83} proposed a further modification of Sterne's method giving intervals with more regular monotonic endpoint sequences than Sterne's and Crow's, while also achieving minimal total width. \citet{Blaker00,Blaker01} proposed an improvement  of the Clopper-Pearson method giving shorter intervals, nested by confidence level, by choosing a more efficient partition of the error probabilities than the ``equal tails'' approach of the earlier method. A referee brought to our attention the recent method of \citet{Schilling14}, which produces length-minimizing, exact intervals for the Binomial problem by shifting acceptance intervals to achieve monotonicity of endpoints before inverting; this is similar to our approach to the hypergeometric.

For the hypergeometric, pivoting the c.d.f.\ was proposed by \citet{Konijn73} and \citet{Buonaccorsi87}, but length-optimality was not addressed until \citet{Wang15}, who proposed a computationally-intensive method for both 1-and 2-sided intervals, and proved that the 1-sided intervals were length-minimizing. See Section~\ref{sec:ex.comp} for a more detailed description and comparison of these methods.

\citet[][p.~463]{Casella02} give a summary of work on confidence sets for some other discrete distributions. One notable example is Crow and Gardner's~\citeyearpar{Crow59} for the Poisson mean, a method similar to Crow's~\citeyearpar{Crow56} for the Binomial.

\section{Additional notation}
Throughout the paper we treat the positive integers $n$ and $N$, and the desired confidence level~$1-\alpha\in(0,1)$, as fixed quantities, known to the statistician, and inference centers on the unknown value of $M$.  Since the parameter~$M$ of interest is an integer, the \textit{intervals} we consider are actually sets of consecutive integers, which we denote by $[a,b]$ but actually mean $\{a,a+1, \ldots,b\}$.  For an arbitrary set~$A$ we let $P_M(A)$ denote $P_M(X\in A)$ where $X\sim\mbox{Hyper}(M,n,N)$, which $X$ will denote throughout unless otherwise specified. For a scalar $x$ we let $P_M(x)$ denote $P_M(X=x)$. We let $\lfloor y\rfloor$ denote the largest integer $\le y$ and $\lceil y\rceil$ the smallest integer $\ge y$. For sets $A, B$ let $A\setminus B=\{a\in A:\; a\notin B\}$ denote the set difference and $|A|$ denote set cardinality, e.g., $|[a,b]|=b-a+1$ for integers $a\le b$. For a nonnegative integer $j$ we let $[j]=\{0,1,\ldots,j\}$.

%\clearpage

\section{\AMO{} acceptance sets and modifying intervals for monotonicity}\label{sec: adj}

In this section we establish properties of acceptance intervals that will guarantee that they still enjoy size optimality when they are appropriately shifted to make their endpoints monotonic.  The next definition makes this precise, and  we call the property \emph{$\alpha$ max optimal}. Theorem~\ref{thm:adj} shows how to modify any set of \AMO{} acceptance intervals to produce intervals whose endpoints $a_M, b_M$ are nondecreasing in $M$, thus producing confidence \emph{intervals} upon inversion rather than non-interval confidence \emph{sets}; see also Section~\ref{sec:length.CI}. It is not difficult to construct \AMO{} acceptance intervals, and a simple and straightforward algorithm to do so which we call Algorithm~\ref{alg:minacc} is given in Section~\ref{sec:alg1.AMO}, where we also prove that it is \AMO{}.

For the next definition we more generally consider acceptance \emph{sets} (not necessarily intervals): A \emph{level-$\alpha$ acceptance set} for $H(M)$ is any subset~$S_M\subseteq[n]$ such that \begin{equation*}
P_M(S_M)\ge 1-\alpha. 
\end{equation*}

\begin{definition} \label{def:a.opt} Fix $n$, $N$, and $\alpha\in(0,1)$.
\begin{enumerate}
\item  Given $M\in [N]$, a subset $S\subseteq [n]$ is \emph{\AO{} for $M$} if $P_M(S)\ge 1-\alpha$ and $P_M(S^*)< 1-\alpha$ whenever $S^*\subseteq [n]$ with $|S^*|<|S|$. A collection $\{S_M:\; M\in\mM\}$, $\mM\subseteq [N]$, is \emph{\AO{} (for $\mM$)} if,  for all $M\in\mM$, $S_M$ is \AO{} for $M$.
\item Given $M\in [N]$, a subset $S\subseteq [n]$ is \emph{$P_M$-maximizing} if all elements of $S$ have positive $P_M$-probability and $P_M(S)\ge P_M(S^*)$ whenever $|S^*|= |S|$. A collection $\{S_M:\; M\in\mM\}$, $\mM\subseteq [N]$, is \emph{$P_{\mM}$-maximizing} if,  for all $M\in\mM$, $S_M$ is $P_M$-maximizing.
\item A collection $\{S_M:\; M\in\mM\}$, $\mM\subseteq [N]$, is \emph{\AMO{} (for $\mM$)} if it is \AO{} and $P_{\mM}$-maximizing.
\end{enumerate}
\end{definition}

The next proposition shows the link between the more general probability-maximizing \emph{sets} in the definition, and intervals; namely, that probability-maximizing sets are always intervals.

\begin{proposition} If $S$ is $P_M$-maximizing for some $M$ then it is a subinterval of $[x_{\min}, x_{\max}]$, as defined in \eqref{x.range}.
\end{proposition}

\begin{proof}
By definition $S\subseteq [x_{\min}, x_{\max}]$, and by the (near-strict) unimodality of every $P_M$ in Lemma~\ref{lem:hyp.uni.x} there exists $[m_1,m_2]\subseteq [x_{\min}, x_{\max}]$ with $m_2-m_1=0$ or $1$ such that $P_M(x)$ is strictly increasing on $[x_{\min},m_1]$ and strictly decreasing on $[m_2,x_{\max}]$. Suppose $x_1,x_2\in S$ and $x_1<y<x_2$. Then $y\le m_1$ implies $P_M(y)>P_M(x_1)$ and $y\ge m_2$ implies $P_M(y)>P_M(x_2)$. Since $m_1$ and $m_2$ are equal or adjacent, at least one of these cases applies, and therefore $P_M(y)>\min\{P_M(x_1), P_M(x_2)\}$. By Lemma~\ref{larger} it follows that $y\in S$ and thus $S$ is an interval.
\end{proof}

Our main result concerning \AMO{} acceptance intervals, stated in the next theorem, is that they can always be modified in order to make both sequences of endpoints nondecreasing in $M$ while still being \AO{}.

\begin{theorem} \label{thm:adj}
Fix $n$, $N$, $\alpha\in(0,1)$. Let $\mM \subseteq[N]$  be an arbitrary set of consecutive integers, and $\{[a_M, b_M]:\; M\in\mM\}$ a set of \AMO{} acceptance intervals. For $M\in\mM$ define
 \begin{equation}\label{oa.ub.def}
\oa_M=\max_{M'\le M}a_{M'}\qmq{and}\ub_M=\min_{M'\ge M} b_{M'}.
\end{equation}
 Finally, define
\begin{equation}\label{Ma.Mb.def}
\mM_a=\{M\in\mM:\; a_M<\oa_M\}\qmq{and}\mM_b=\{M\in\mM:\; b_M>\ub_M\}.
\end{equation}
Then the following hold.
\begin{enumerate}
\item\label{thm.adj.Mab} The sets $\mM_a$ and $\mM_b$ are disjoint.
\item\label{thm.adj.ints}  The adjusted intervals
\begin{equation}\label{thm.adj.ints.def}
[a_M^{adj},b_M^{adj}]:=\begin{cases}
[\oa_M,b_M+(\oa_M-a_M)],&M\in\mM_a\\
[a_M-(b_M-\ub_M),\ub_M],&M\in\mM_b\\
[a_M,b_M],&\mbox{all other $M\in\mM$,}
\end{cases} 
\end{equation}
are \AO{} and have nondecreasing endpoint sequences.
\end{enumerate}
\end{theorem}

The proof of the theorem and auxiliary results are given in Appendix~\ref{sec:a.max.opt.proofs}.

\section{\AO{}, symmetrical, nondecreasing acceptance intervals}\label{sec:mod.minacc}

We  first present Algorithm~\ref{alg:minacc} and verify that it is \AMO{} in Section~\ref{sec:alg1.AMO}. This sets the stage to apply Theorem~\ref{thm:adj} to the Algorithm~\ref{alg:minacc} acceptance intervals
\begin{equation}\label{alg1.N/2}
\{[a_M,b_M]:\; M\in [\lfloor N/2\rfloor]\}
\end{equation}
 to obtain the modified intervals
\begin{equation}\label{alg1.adj.N/2}
\{[a_M^{adj},b_M^{adj}]:\; M\in [\lfloor N/2\rfloor]\}.
\end{equation}

A set of acceptance intervals $\{[a_M,b_M]: M\in[N]\}$ is \textit{symmetrical} if the intervals are equivariant with respect to the reflections~$M\mapsto N-M$, $[a_M,b_M]\mapsto [n-b_M,n-a_M]$.  That is, if
\begin{equation}\label{acc.adm.def}
[a_{N-M},b_{N-M}] = [n-b_M, n-a_M]\qm{for all $M \in[N]$.}
\end{equation}
This can equivalently be stated as the intervals $[a_M,b_{N-M}]$, $M\in[N]$, all having midpoint $n/2$, or having endpoints summing to $n$. We seek symmetrical acceptance intervals because they will result in symmetrical confidence intervals (defined below analogously to \eqref{acc.adm.def}) upon inversion in Section~\ref{sec:length.CI}. From \eqref{alg1.adj.N/2}, one way to achieve symmetry is to define $[a_M^{adj},b_M^{adj}]$ for $M>\lfloor N/2\rfloor$ as the reflection of the intervals~\eqref{alg1.adj.N/2} across $n/2$. This achieves symmetry everywhere except at $M=N/2$ when $N$ is even and $[a_{N/2}^{adj},b_{N/2}^{adj}]$ is not symmetric about $n/2$.  This is the strategy taken in Theorem~\ref{thm:a*b*}, with the $M=N/2$ interval taken to be \eqref{N/2.sym.int}, and the resulting intervals are \AO{}, symmetrical, and have nondecreasing endpoint sequences. We call the result of applying Theorem~\ref{thm:a*b*} to Algorithm~\ref{alg:minacc}, Algorithm~\ref{alg:AO.ints}, which is also given in algorithmic form in Appendix~\ref{sec:alg2}.  However, Theorem~\ref{thm:a*b*} is general and applies to not just Algorithm~\ref{alg:minacc} but starting from any \AMO{} intervals~\eqref{alg1.N/2}.

 \subsection{\AMO{} intervals: Algorithm~\ref{alg:minacc}}\label{sec:alg1.AMO}
 
The following is a simple algorithm for producing \AMO{} acceptance intervals. The algorithm starts from the endpoints $C=D=\arg\max_x P_M (x)$ equal to the mode, and moves the endpoints outward, incrementing the acceptance probability.  An alternative is to begin from $[C,D]=[0, N]$ and move the endpoints inward, decrementing the probability, although this will be slower in most settings.

\begin{algorithm}[!htp]
\caption{Given $\alpha$, $n$, and $N$, produce a set of level~$\alpha$ acceptance intervals~$\{[a_M,b_M]:\; M\in[N]\}$.}

\begin{algorithmic} 
\REQUIRE $N \in \mathbb{N}$, $n \leq N$ and $ 0 < \alpha < 1$

\FOR{$M = 0, ..., \lfloor N/2 \rfloor$ }
\STATE $x_{\min}= \max\{0,M+n-N\}$
\STATE $x_{\max}= \min\{n,N\}$
\STATE $C, D = \lfloor \frac{(n+1)(M+1)}{N+2}\rfloor$
\STATE $P = P_M(C)$
\STATE\textbf{if}  $C>x_{\min}$ \textbf{then} $PC= P_M(C-1)$ \textbf{else} $PC= 0$ \textbf{end if}
\STATE\textbf{if}  $D<x_{\max}$ \textbf{then} $PD = P_M(D+1)$ \textbf{else} $PD= 0$ \textbf{end if}
\WHILE{$P< 1 - \alpha$}
\IF{$PD>PC$}
\STATE $D= D+1$
\STATE $P= P+PD$
\STATE\textbf{if}  $D<x_{\max}$ \textbf{then} $PD = P_M(D+1)$ \textbf{else} $PD= 0$ \textbf{end if}
\ELSE 
\STATE $C= C-1$
\STATE $P= P+PC$
\STATE\textbf{if}  $C>x_{\min}$ \textbf{then} $PC= P_M(C-1)$ \textbf{else} $PC= 0$ \textbf{end if}
\ENDIF
\ENDWHILE
\STATE $a_M = C$
\STATE $b_M = D$
\STATE $a_{N - M} = n - b_M$
\STATE $b_{N - M} = n - a_M$
\ENDFOR
\RETURN $\{[a_M,b_M]\}^{N}_{M = 0}$ 
\end{algorithmic}
\label{alg:minacc}
\end{algorithm}

\begin{lemma}\label{lem:alg1.AMO}
The acceptance intervals $\{[a_M, b_M]:\; M\in [\lfloor N/2\rfloor]\}$ produced by Algorithm~\ref{alg:minacc} are \AMO{}.
\end{lemma}

\begin{proof} Fix $M\in [\lfloor N/2\rfloor]$ and use $P$ to denote $P_M$. Recall that $P(x)=0$ for $x$ outside $[x_{\min},x_{\max}]$. Algorithm~\ref{alg:minacc} builds up a nested sequence of acceptance intervals $I,J,K,\ldots$ by selecting at each stage a new point having maximum probability among the ones available. The first interval is $I=\{\lfloor m\rfloor\}$ where $m$ is given by \eqref{argmax.x} and $\lfloor m\rfloor$, and thus $I$, have maximum probability; see Lemma~\ref{lem:hyp.uni.x}.  At each successive stage, the current interval $J$ is expanded to $K=J\cup \{x\}$, where  $x$ is $y$, the point adjacent to $J$ on the left, or $z$, adjacent on the right, chosen to yield the maximum of their probabilities.  Since $y$ is to the left of $\{\lfloor m\rfloor\}$, the points further to the left have smaller probabilities by virtue of the unimodality property in Lemma~\ref{lem:hyp.uni.x}, and similarly $z$ has the maximum probability to the right.  Therefore the selected point $x$ has the maximum probability outside of the current interval~$J$, and hence by the Definition~\ref{def:a.opt} of $P$-maximizing sets, the next interval $K$ preserves the $P$-maximizing property of $J$, using an obvious induction argument.
 
For the given $M$, the while loop in Algorithm~\ref{alg:minacc} ends with the first interval $V$ that is level $\alpha$. Hence its predecessor is not level $\alpha$, and since it is $P$-maximizing, no set with fewer points than $V$ is level $\alpha$.  Thus $V$ is \AO{} as well as $P$-maximizing, and hence it is \AMO{}.
\end{proof}

\subsection{Reflection and modification at $N/2$: Algorithm~\ref{alg:AO.ints}}\label{sec:alg2.intro}

Starting with a set of \AMO{} intervals~$\{[a_M,b_M]:\; M\in [\lfloor N/2\rfloor]\}$ we will now define a new set of intervals~$\{[a^*_M, b^*_M]:\; M\in [N]\}$ by (i) applying the adjustments in Theorem~\ref{thm:adj}, (ii) reflecting across $n/2$ to obtain symmetrical intervals for $M>\lfloor N/2\rfloor$, and (iii) if $N$ is even setting $[a^*_{N/2}, b^*_{N/2}]$ to be the interval 
\begin{equation}\label{N/2.sym.int}
[h_{\alpha/2}, n - h_{\alpha/2}],\qmq{where} h_{\alpha/2} = \max\left\{x \in [n] :\; P_{N/2}(X < x ) \leq \alpha/2\right\}.
\end{equation}
The next theorem establishes that the  resulting intervals are \AO{}, symmetrical, and have nondecreasing endpoint sequences.

\begin{theorem}\label{thm:a*b*} Given a set of \AMO{} intervals~$\{[a_M,b_M]:\; M\in [\lfloor N/2\rfloor]\}$, let $\{[a_M^{adj},b_M^{adj}]:\; M\in [\lfloor N/2\rfloor]\}$ denote the result of applying Theorem~\ref{thm:adj}, and  
\begin{equation}\label{a*b*.def}
[a^*_M, b^*_M]=\begin{cases}
[a^{adj}_M, b^{adj}_M],&\mbox{for $M=0,1,\ldots, \lceil N/2\rceil -1$;}\\
[n-b^{adj}_{N-M}, n-a^{adj}_{N-M}],&\mbox{for $M= \lfloor N/2\rfloor +1,\ldots, N$;}\\
[h_{\alpha/2}, n - h_{\alpha/2}],&\mbox{for $M=N/2$ if $N$ is even.}
\end{cases}
\end{equation}
Then $\mA^*:=\{[a_M^*,b_M^*]:\; M\in [N]\}$ are level-$\alpha$, symmetrical,  have nondecreasing endpoint sequences, and are size-optimal except possibly for $M=N/2$ when $N$ is even; in this case, $[a_{N/2}^*,b_{N/2}^*]$ is size-optimal unless $[a_{N/2}^*,b_{N/2}^*-1]$ has probability $1-\alpha$ or greater, in which case $[a_{N/2}^*,b_{N/2}^*-1]$ is \AO{} and $\mA^*$ is larger by one in total size than an \AO{} collection. In any case, $\mA^*$ is \AO{} among symmetrical collections.
\end{theorem}

Note that if $N$ is odd then the first two cases of \eqref{a*b*.def} cover all $M\in[N]$.

Before the proof of the theorem, we make a few comments about the $M=N/2$ interval~\eqref{N/2.sym.int} when $N$ is even. It is clearly the smallest symmetrical level-$\alpha$ acceptance interval for $H(N/2)$, and is \AO{}. Also, we have $h_{\alpha/2} \leq n/2$ since 
\begin{equation}\label{h<n/2}
P_{N/2}(X < \lfloor n/2 \rfloor + 1) = P_{N/2}(X \leq \lfloor n/2 \rfloor) \geq 1/2 > \alpha/2.
\end{equation}
 Finally, it is not necessary to use special calculations to get $h_{\alpha/2}$ since it is easily obtained from an \AMO{} interval $[a_{N/2}, b_{N/2}]$ by $h_{\alpha/2}=\min\{a_{N/2},n-b_{N/2}\}$. Note that if $[a_{N/2}, b_{N/2}]$ is already symmetrical, then \eqref{N/2.sym.int} is the same interval.

\begin{proof} Symmetry is by construction, and monotonicity is proved in Lemma~\ref{lem:a*b*.mono}. Theorem~\ref{thm:adj} establishes $\alpha$-optimality in the first case of \eqref{a*b*.def}, and to establish it in the second case, let $M>\lfloor N/2\rfloor$ and $X\sim\mbox{Hyper}(M,n,N)$. Then 
\begin{multline*}
P_M(X\in[a_M^*,b_M^*]) = P_M(X\in[n-b_{N-M}^*, n-a_{N-M}^*]) \\
= P_M(n-X\in[a_{N-M}^*, b_{N-M}^*]) \ge 1-\alpha,
\end{multline*}
this last by the first case of \eqref{a*b*.def} and since $n-X\sim\mbox{Hyper}(N-M,n,N)$; see Lemma~\ref{lem:hyp.basics}. For size optimality, we will show that $P_M([c,d])\ge 1-\alpha$ implies that $d-c\ge b_M^*-a_M^*$. By an argument similar to the one above, $[n-d,n-c]$ is level-$\alpha$ for testing $H(N-M)$ and thus no shorter than $[a_{N-M}^*,b_{N-M}^*]$, so
$$d-c=(n-c)-(n-d)\ge b_{N-M}^*-a_{N-M}^* = (n-a_M^*) - (n-b_M^*) = b_M^*-a_M^*,$$ as claimed. For $N$ even, the third case of \eqref{a*b*.def} is clearly level-$\alpha$, and any competing symmetrical interval for $H(N/2)$ must be of the form $[c,n-c]$. If this is level-$\alpha$ then $c\le h_{\alpha/2}$, thus it can be no shorter than \eqref{N/2.sym.int}.
\end{proof}

\section{Optimal symmetrical confidence intervals}\label{sec:length.CI}
\subsection{Confidence and acceptance sets}
For a set~$\mS$ let $2^{\mS}$ denote the power set of $\mS$, i.e., the set of all subsets of $\mS$. A \textit{confidence set with confidence level $1-\alpha$} is a function $\mC:[n]\To 2^{[N]}$ such that the coverage probability satisfies
\begin{equation*}
P_M(M\in\mC(X))\ge 1-\alpha\qmq{for all}M\in[N].
\end{equation*}
For short, we refer to such a $\mC$ as a $(1-\alpha)$-confidence set. If a confidence set~$\mC$ is interval-valued (i.e., for all $x\in[n]$, $\mC(x)$ is an interval) we call it a \textit{confidence interval}. A confidence set $\mC$ is \textit{symmetrical} if
\begin{equation}\label{CI.adm.def}
\mC(x)=N-\mC(n-x) \qm{for all $x\in[n]$.}
\end{equation} Here, for a set~$\mS$, the notation $N-\mS$ means $\{N-s:\; s\in\mS\}$. Symmetry~\eqref{CI.adm.def} is an equivariance condition requiring that the confidence set is reflected about $N/2$ when the data is reflected about $n/2$. See also Section~\ref{sec:ex.comp} for how this definition compares with the regularity conditions of W.~Wang~\citeyearpar{Wang15}.

Similarly, we shall denote a level-$\alpha$ acceptance set by a function $\mA: [N]\To 2^{[n]}$ such that
\begin{equation*}
P_M(X\in\mA(M))\ge 1-\alpha\qmq{for all}M\in[N],
\end{equation*}
and call an interval-valued (i.e., for all $M\in[N]$, $\mA(M)$ is an interval) acceptance set an acceptance interval\footnote{Note that whereas above we referred to an expression like \eqref{alg1.N/2} as a set of acceptance intervals, we will now call it \emph{an} acceptance interval (singular).  This is to coincide with our terminology for a confidence set, as well as avoid cumbersome phrases like ``a set of acceptance sets.''} and write $\mA(M)=[a_M,b_M]$, or similar.

We also need to generalize the concept of symmetry from \eqref{acc.adm.def} to handle general sets, so we say that an acceptance set~$\mA$ is \textit{symmetrical} if 
\begin{equation*}
\mA(M) = n - \mA(N-M)\qmq{for all}M\in[N].
\end{equation*} 
This says that the set is equivariant with respect to reflections $M\mapsto N-M$, and specializes to \eqref{acc.adm.def} for intervals.

\subsection{Inverted confidence sets}

We will construct confidence sets that are inversions of acceptance sets, and vice-versa. If $\mA$ is a level-$\alpha$ acceptance set then
\begin{equation}\label{CS.gen}
\mC_{\mA}(x)=\{M\in[N]:\; x\in\mA(M)\}
\end{equation} is a $(1-\alpha)$-confidence set.  Conversely, given a $(1-\alpha)$-confidence set~$\mC$,
\begin{equation*}
\mA_{\mC}(M)=\{x\in[n]:\; M\in\mC(x)\}
\end{equation*} is a level-$\alpha$ acceptance set; see, for example, \citet[][Chapter~9.3]{Rice07}. Moreover, $\mC_{\mA_{\mC}} =\mC$ and $\mA_{\mC_{\mA}} = \mA$, which are immediate from the definitions. However, neither $\mA$ nor $\mC$ being interval-valued guarantees that its inversion is. 

We will evaluate confidence and acceptance sets by their \textit{total size}, which we define as the sum of the cardinalities of each set: Recalling that $|\cdot|$ denotes set cardinality, define the \textit{total size} of acceptance and confidence sets to be
\begin{equation*}
|\mA|=\sum_{M=0}^N|\mA(M)| \qmq{and}|\mC| = \sum_{x=0}^n |\mC(x)|.
\end{equation*} 
If $\mA(M)=[a_M,b_M]$ is an acceptance interval then
\begin{equation*}
|\mA|=\sum_{M=0}^N|[a_M,b_M]|=\sum_{M=0}^N (b_M-a_M+1),
\end{equation*}  and similarly for a confidence interval~$\mC$.

Lemma~\ref{big theorem} records some basic facts about inverted confidence sets.  
 
\begin{lemma}\label{big theorem}
Let $\mA$ be an acceptance set. Then the following hold.
\begin{enumerate}
\item\label{lem.part:tot.len}  
\begin{equation}\label{lenCA=lenA}
|\mC_{\mA}| = |\mA|.
\end{equation}
\item\label{lem.part:Aad.Cad} $\mC_{\mA}$ is symmetrical if and only if $\mA$ is symmetrical.
\item\label{lem.part:Amn.Cint} If, in addition, $\mA(M)=[a_M,b_M]$ is interval-valued and the endpoint sequences $\{a_M\}$ and $\{b_M\}$ are nondecreasing, then $\mC_{\mA}$ is interval-valued.
\end{enumerate}
\end{lemma}

\begin{proof} Denote $\mC_{\mA}$ simply by $\mC$. For part~\ref{lem.part:tot.len}, letting $\bm{1}\{\cdot\}$ denote the indicator function,
\begin{multline*}
|\mC|=\sum_{x\in[n]}|\mC(x)|=\sum_{x\in[n]}|\{M\in[N]: x\in\mA(M)\}|\\
=\sum_{x\in[n],\; M\in[N]}\bm{1}\{x\in\mA(M)\}=\sum_{M\in[N]}|\{x\in[n]: x\in\mA(M)\}|\\
=\sum_{M\in[N]}|\mA(M)| = |\mA|.
\end{multline*}

For part~\ref{lem.part:Aad.Cad}, if $\mA$ is symmetrical,
\begin{align*}
\mC(x)&=\{M\in[N]:\; x\in\mA(M)\}\\
&= \{M\in[N]:\; x\in n-\mA(N-M)\}\\
&= \{M\in[N]:\; n-x\in\mA(N-M)\}\\
&= \{N-M\in[N]:\; n-x\in\mA(M)\}\\
&= N-\{M\in[N]:\; n-x\in\mA(M)\}\\
&=N-\mC(n-x).
\end{align*} A similar argument shows the converse.

For part~\ref{lem.part:Amn.Cint}, fix arbitrary $x\in[n]$ and to  show that $\mC(x)$ is an interval, suppose that $M_1, M_2\in\mC(x)$ and we will show that $M\in\mC(X)$ for all $M_1<M<M_2$. Since $M_2\in\mC(x)$, $x\in[a_{M_2}, b_{M_2}]$ so $x\ge a_{M_2}\ge a_M$ by monotonicity.  By a similar argument, $x\le b_{M_1}\le b_M$, thus $x\in[a_M,b_M]$ so $M\in\mC(x)$.

\end{proof}

\subsection{Size optimality}\label{sec:C*.sz.opt}
We say that a confidence set~$\mC$ is \textit{size-optimal} among a collection of confidence sets if it achieves the minimum total size in that collection. The results in this section establish size-optimality of $\mC^*=\mC_{\mA^*}$, where  $\mA^*=\{[a_M^*,b_M^*]:\; M\in[N]\}$ denotes the result of applying Theorem~\ref{thm:a*b*} to any \AMO{} acceptance intervals~$\{[a_M,b_M]:\; M\in[N]\}$. Thus, $\mA^*$ could be the intervals given by Algorithm~\ref{alg:AO.ints}, or the result of starting with any other \AMO{}  intervals. Whatever the choice of $\mA^*$, note that $\mC^*$ is a symmetrical, $(1-\alpha)$-confidence interval by Lemma~\ref{big theorem}

Theorems~\ref{thm:len.opt.set} and \ref{thm:len.opt.int}, which follow, are the main results of the paper.  Theorem~\ref{thm:len.opt.set} is the more powerful of the two in that it gives wide conditions under which $\mC^*$ is size-optimal among symmetrical confidence \emph{sets} (not just intervals) and shows that, even in the worst case, the total size $|\mC^*|$ is at most $1$ point larger than the optimal set. Theorem~\ref{thm:len.opt.int} specializes to intervals and gives conditions for optimality there. In particular, it shows that $\mC^*$ is size-optimal among all symmetrical non-empty (i.e., $\mC(x)\ne\emptyset$ for all $x$) intervals, which are usually preferred in practice.  

\begin{theorem}\label{thm:len.opt.set}
Let $\mC^*$ be as defined above and $\fC_S$ the class of all symmetrical, $(1-\alpha)$-confidence sets. Then $\mC^*$ is size-optimal in $\fC_S$, i.e.,
\begin{equation*}
|\mC^*|= \min_{\mC\in\fC_S} |\mC|,
\end{equation*}
if either of the following holds:
\begin{enumerate}[label= (\alph*)]
\item $n$ or $N$ is odd;
\item\label{part:set.nonopt} $n,N$ are even and there is no size-optimal $\mC\in\fC_S$ such that $|\mA_{\mC}(N/2)|$ is even. If $n$, $N$, and  $|\mA_{\mC}(N/2)|$ are all even for some size-optimal $\mC\in\fC_S$, then 
\begin{equation}\label{le.set.opt+1}
|\mC^*|\le \min_{\mC\in\fC_S} |\mC|+1.
\end{equation}
\end{enumerate}
In addition, $\mC^*$ is size-optimal among all $\mC\in\fC_S$ such that $\mA_{\mC}$ are all intervals.
\end{theorem}

The proofs of Theorems~\ref{thm:len.opt.set}and \ref{thm:len.opt.int} utilize some auxiliary lemmas, stated and proved in Appendix~\ref{sec:len.opt.aux}. See also Example~\ref{ex:AN/2.even} for an instance of $\mC^*$ failing to be optimal under conditions satisfying part~\ref{part:set.nonopt}.

\begin{proof}[Proof of Theorem~\ref{thm:len.opt.set}]
First suppose $N$ is odd, and let $\mC\in\fC_S$ be arbitrary;  we will show that $|\mC^*|\le|\mC|$.  Since $N/2$ is not an integer, by Lemma~\ref{lem:A*<ACM} we have $|\mA^*(M)|\le|\mA_{\mC}(M)|$ for all $M\in[N]$ thus, using Lemma~\ref{big theorem},
\begin{equation}\label{C>C*.set}
|\mC|=|\mA_{\mC}|=\sum_{M=0}^N |\mA_{\mC}(M)|\ge \sum_{M=0}^N |\mA^*(M)| = |\mA^*|=|\mC^*|,
\end{equation} as claimed.

Now suppose $N$ is even and let $\mC\in\fC_S$ be size-optimal.   By Lemma~\ref{lem:A*<ACM} we have $|\mA^*(M)|\le|\mA_{\mC}(M)|$ for all $M\in[N]$ other than $M=N/2$. If $n$ or $|\mA_{\mC}(N/2)|$ is odd, then by  Lemma~\ref{lemma: interval 2 } there is an interval~$[a, n-a]$ such that $n - 2a + 1 = |\mA_{\mC}(N/2)|$  and $P_{N/2}( [a, n - a]) \geq P_{N/2}( \mA_{\mC}(N/2))$. Since $\mA^*(N/2)=[a_{N/2}^*, b_{N/2}^*] = [a_{N/2}^*, n  - a_{N/2}^*]$ is the shortest symmetrical acceptance interval for $M = N/2$, we have  
\begin{equation*}
|\mA^*(N/2)|=b_{N/2}^* - a_{N/2}^*+1\le n-2a +1= |\mA_{\mC}(N/2)|.
\end{equation*} This, with the above inequality for the $M\ne N/2$ cases, establishes \eqref{C>C*.set} in this case.

The remaining case -- when $N$, $n$, and $|\mA_{\mC}(N/2)|$ are all even -- is handled by Lemma~\ref{lem:Aodd.C*opt}, recalling that $\mC$ was size-optimal to establish \eqref{le.set.opt+1}.

For the final statement in the theorem, for any such $\mC$, $\mA_{\mC}$ is symmetrical and thus has total size at least $|\mA^*|$, so $|\mC|=|\mA_{\mC}|\ge  |\mA^*|=|\mC^*|$.
\end{proof}

\begin{theorem}\label{thm:len.opt.int}
Let $\mC^*$ be as defined above and $\fC_I$ the class of all symmetrical, $(1-\alpha)$-confidence intervals. Then $\mC^*$ is size-optimal in $\fC_I$, i.e.,
\begin{equation*}
|\mC^*|= \min_{\mC\in\fC_I} |\mC|,
\end{equation*}
if either of the following holds:
\begin{enumerate}[label= (\alph*)]
\item\label{thm.part:len.opt.odd.int} $n$ or $N$ is odd;
\item\label{thm.part:len.opt.empty} $n,N$ are even and there is no size-optimal $\mC\in\fC_I$ such that 
\begin{equation}\label{C'=empty}
\mC(n/2)=\varnothing.
\end{equation}  A sufficient condition for $\mC^*$ to be size-optimal in this case is that 
\begin{equation}\label{suff.CI.opt}
\alpha< \left.{N/2\choose n/2}^2\right/{N\choose n}.
\end{equation}

\end{enumerate}
 In particular, $\mC^*$ is size-optimal  among all nonempty $\mC\in\fC_I$ regardless of the parity of $n, N$.
\end{theorem}

We comment that the scenario~\eqref{C'=empty} seems to be particularly rare since $\mC(n/2)$ is typically the widest confidence interval.  Thus, even allowing empty intervals, Theorem~\ref{thm:len.opt.int} establishes size optimality of $\mC^*$ among intervals for most intents and purposes, and \eqref{le.set.opt+1} holds in any case.  However, it may be possible to construct an adversarial example with that property.

\begin{proof}[Proof of Theorem~\ref{thm:len.opt.int}] Part \ref{thm.part:len.opt.odd.int} is a consequence of Theorem~\ref{thm:len.opt.set} since $\fC_I\subseteq \fC_S$.

Assume $N$ and $n$ are even, and there is no size-optimal $\mC$ satisfying \eqref{C'=empty}. Let $\mC$ be any size-optimal interval and since $\mC(n/2) \neq \varnothing$, there is some $M \in \mC(n/2)$. Because $\mC(n/2)$ is symmetrical, $N - M \in \mC(n/2)$, and because $\mC(n/2)$ is an interval, $N/2 \in \mC(n/2)$ since it lies between $M$ and $N - M$. This implies that $n/2\in\mA_{\mC}(N/2)$, which is symmetrical about $n/2$.  Using these facts,
\begin{align*}
|\mA_{\mC}(N/2)| &= 2|\{x \in \mA_{\mC}(N/2) \mid x < n/2\}| + |\{n/2\}| \\
&= 2|\{x \in \mA_{\mC}(N/2) \mid x < n/2\}| + 1,
\end{align*}
an odd number. We then have $|\mC^*|\le |\mC|$ by Lemma~\ref{lem:Aodd.C*opt}.

To see that \eqref{suff.CI.opt} is sufficient, suppose there is a $\mC$ with $\mC(n/2)= \varnothing$.  Then for $M=N/2$, we have 
\begin{equation*}
\alpha\ge P_M(M\not\in\mC(X))\ge P_M(X=n/2)=\left.{N/2\choose n/2}^2\right/{N\choose n}.
\end{equation*} 
\end{proof}

\section{Examples and comparisons}\label{sec:ex.comp}
In this section we show examples of  our proposed method~$\mC^*$ using Algorithm~\ref{alg:AO.ints} as the acceptance interval~$\mA^*$,  and give some comparisons with other methods. All calculations of our method were performed using the \textsf{R} package \texttt{hyperMCI}, available at \url{github.com/bartroff792/hyper}.

For comparisons we focus on \textit{exact} methods with guaranteed coverage probability. A standard method for producing a $(1-\alpha)$-confidence interval for $M$ is the so-called \emph{method of pivoting the c.d.f.}\footnote{We have heard this method alternatively called the \emph{quantile method} and the \emph{method of extreme tails}.} giving $\mC_{Piv}(x) = [L_{Piv}(x),U_{Piv}(x)]$ where, for fixed nonnegative $\alpha_1+\alpha_2=\alpha$,  
\begin{align*}
L_{Piv}(x)&=\min\{M\in[N]: P_M(X\ge x)> \alpha_1\},\\
U_{Piv}(x)&=\max\{M\in[N]: P_M(X\le x)> \alpha_2\}.
\end{align*}
See \citet{Buonaccorsi87}, \citet[][Chapter~9]{Casella02}, or \citet{Konijn73}. Taking $\alpha_1=\alpha_2=\alpha/2$ is a common choice, and all our calculations of $\mC_{Piv}$ below use this.

W.~Wang~\citeyearpar{Wang15} proposed a method producing a $(1-\alpha)$-confidence interval for $M$, which we denote by $\mC_W$, that cycles through the intervals $\mC_{Piv}(x)$, shrinking the intervals where possible while checking that coverage probability is maintained. The algorithm can require multiple passes through the intervals, calculating the coverage  probability for all $M\in[N]$ multiple times, and is therefore computationally intensive. We compare the computational times of $\mC_W$ and $\mC^*$ in Examples~\ref{ex: time and length} and \ref{ex:CW.time}.  All calculations of $\mC_W(x)$ were performed using that author's \texttt{R} code.

 Although W.~Wang proves that a 1-sided version of his algorithm produces size-optimal intervals (among 1-sided intervals), it is not claimed that $\mC_W$ is size-optimal. Since $\mC_W$ produces nonempty intervals we know that $|\mC^*|\le |\mC_W|$ by Theorem~\ref{thm:len.opt.int}. In the following example we compare $\mC^*$ with $\mC_W$ in terms of both size and computational time, and  indeed exhibit a setting where $|\mC^*|<|\mC_W|$. We also note that the regularity conditions assumed in W.~Wang's results are  slightly more restrictive than our symmetry condition~\eqref{CI.adm.def}, which W.~Wang calls a ``natural restriction,''  by including two additional requirements that
 both sequences of endpoints of $\mC_W(x)$ be nondecreasing in $x$, and any sub-interval of $\mC_W(x)$ must have confidence level strictly less than $1-\alpha$. Our $\mC^*$ satisfies these additional properties, and see also Figure~\ref{ci.us.500.100.05} for an example of monotonicity of $\mC^*$. However,  we do not require them of the confidence sets considered so that our optimality results apply to a broader class. 

\begin{ex}\label{ex: time and length} We compare $\mC^*$, $\mC_{Piv}$, and $\mC_W$ in the setting $\alpha = 0.05$, $N=500$, and $n=10,20,30, \ldots, 490$.   Figure~\ref{ci.us.500.100.05} plots the $\mC^*$ intervals as vertical bars for the $n=100$ case.    The $\mC^*$ intervals are much shorter than the $\mC_{Piv}$ intervals in this setting, and Figure~\ref{length.us.pivot}  shows the differences in size $|\mC_{Piv}| - |\mC^*|$ for $n=10, 20,\ldots,490$ which are substantial; all the $\mC^*$ intervals are at least $200$ points shorter than their corresponding $\mC_{Piv}$ intervals, and some are as many as $260$ points shorter.  These differences are also sizable fractions of the largest possible range $[0,N] = [0,500]$.

\begin{figure}[!ht]
    \centering
    \caption{Confidence intervals $\mC^*(x)$ for $\alpha = 0.05$, $N = 500$, $n = 100$, and $x=0,1,\ldots,100$.} \label{ci.us.500.100.05}
\includegraphics[width=14cm]{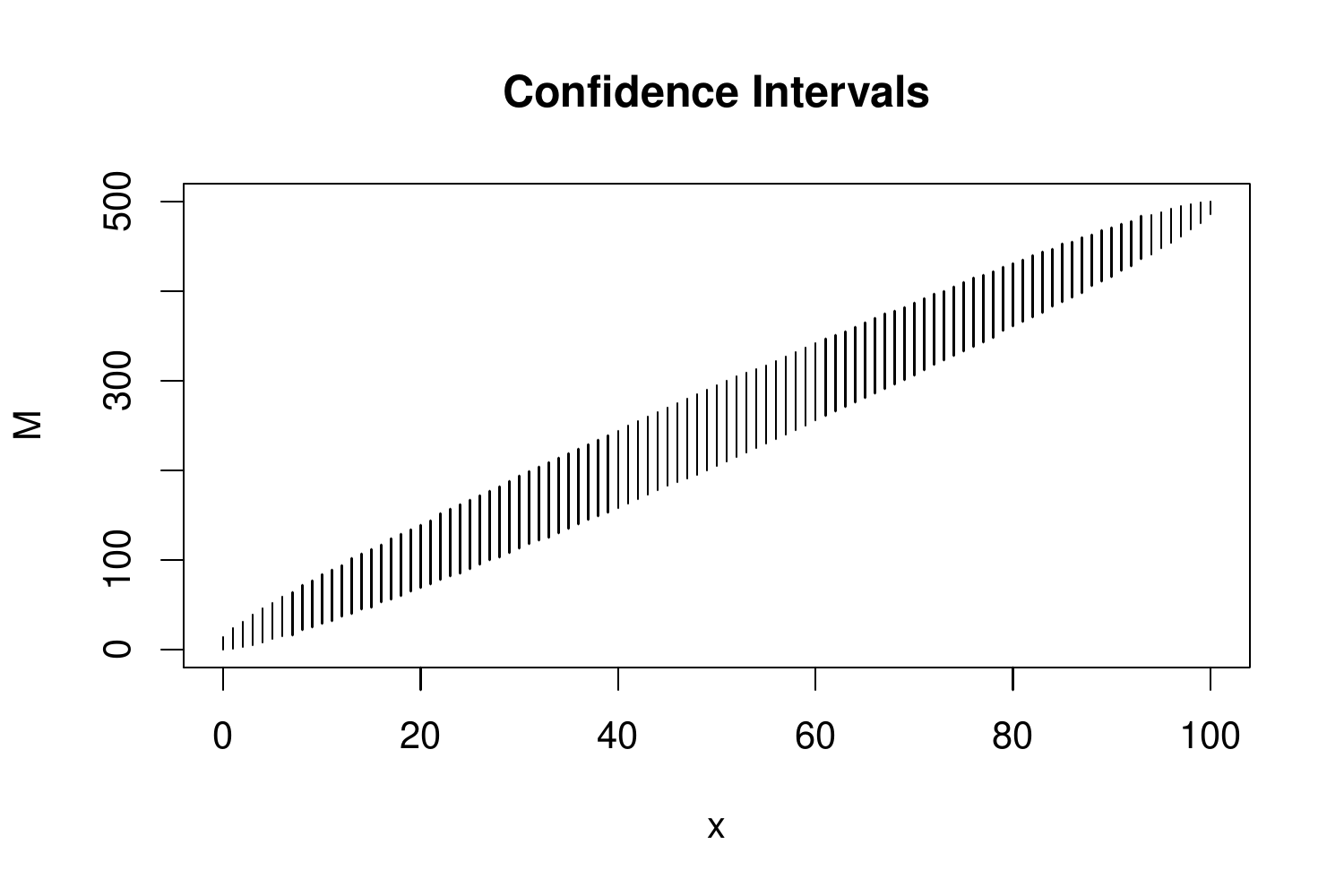} 
\end{figure}

\begin{figure}[!ht]
    \centering
    \caption{The differences in total size $|\mC_{Piv}| - |\mC^*|$, for $N = 500$, $\alpha = 0.05$, and $n=10, 20,\ldots,490$.} \label{length.us.pivot}
\includegraphics[width=12cm]{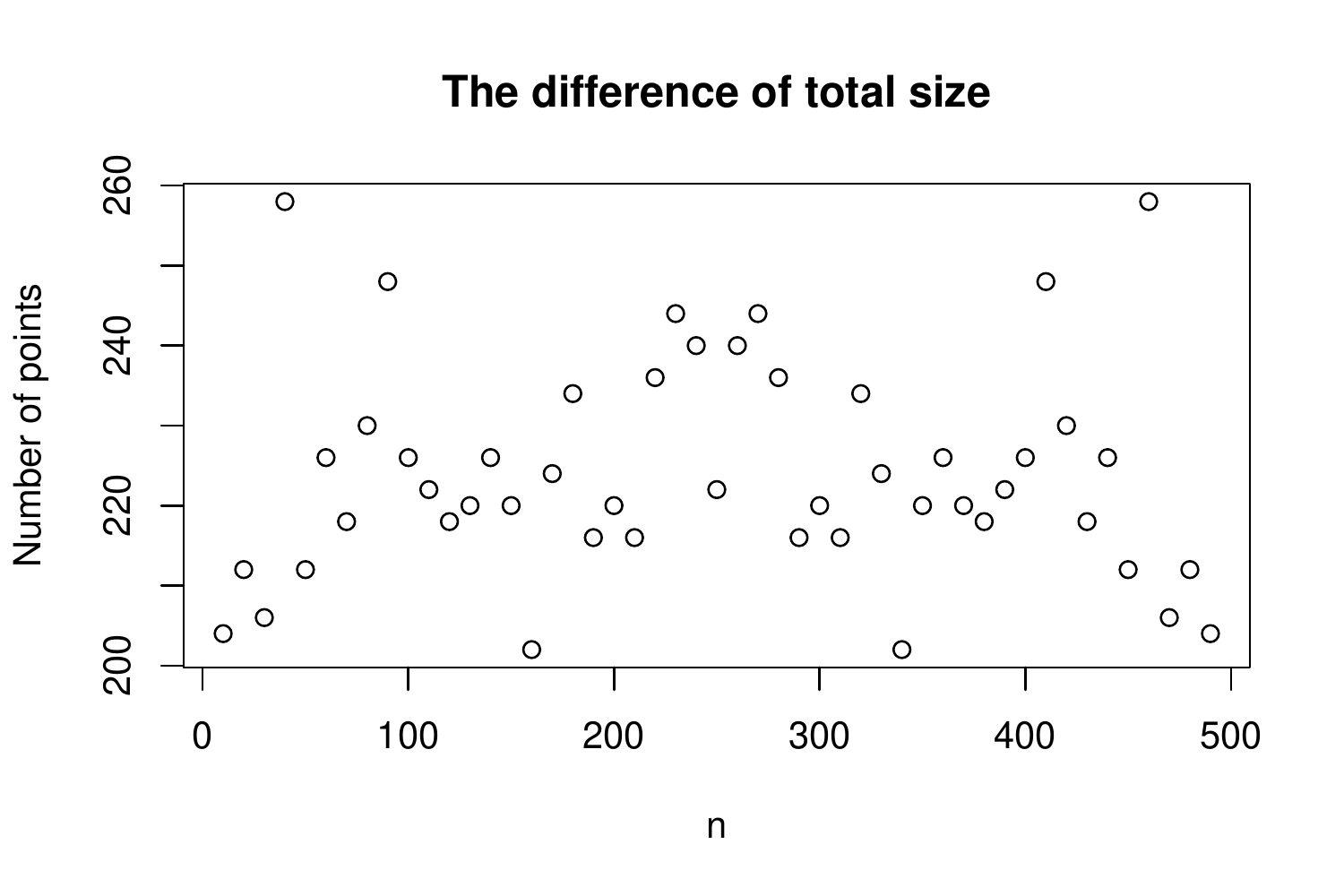} 
\end{figure}

The $\mC_W$ intervals are very similar to $\mC^*$ and so are not shown in Figures~\ref{ci.us.500.100.05}-\ref{length.us.pivot}. In fact, the sizes $|\mC_W|=|\mC^*|$ are exactly equal for all values of $n$ considered, except $n=100$.  To examine this case more closely, we give these confidence intervals explicitly in Tables~\ref{table:us.ne.Wang.500.100} and \ref{table:Wang.ne.us.500.100}. Examining these tables shows very similar, but slightly different intervals, with neither method dominating the other.  For example, $|\mC^*(0)|=|[0,14]|<|[0,16]| = |\mC_W(0)|$ and 
\begin{equation}\label{ex.13}
|\mC^*(13)|=|[40,102]|>|[40,101]| = |\mC_W(13)|. 
\end{equation}
Totaling the sizes gives $|\mC^*|=7129<7131 = |\mC_W|$, showing that the $\mC_W$ intervals are indeed non-optimal. One property of our method is that it does not necessarily producing intervals that are sub-intervals of $\mC_{Piv}$, which $\mC_W$ always does since it begins with these intervals before iteratively shrinking them. For example, in this setting $\mC_{Piv}(13)=[39,101]$ which, by \eqref{ex.13}, contains $\mC_W(13)$ but not $\mC^*(13)$.

\begin{table}[!ht]
\caption{Confidence intervals given by $\mC^*(x)=[L(x),U(x)]$ for $\alpha = 0.05$, $N = 500$, $n = 100$, and $x=0,1,\ldots,100$.}
\begin{tabular}{l|llllllllllllll}
\hline
\hline
$x$  & 0 & 1 & 2  & 3  & 4  & 5  & 6  & 7 & 8 & 9 & 10 & 11  & 12 & 13  \\ \hline
$L(x)$ &0&1&3&5&8&12&15&16&22&25&29&32&37&40\\
$U(x)$ &14&24&31&39&46&52&59&64&72&77&84&89&94&102\\
\hline
$x$  & 14& 15 &16&17&18&19&20&21&22&23&24&25&26&27 \\ \hline
$L(x)$ &45&47 &53&56&60&65&69&73&78&82&85&90&95&100\\
$U(x)$ &107&112 &117&124&129&134&139&144&152&157&162&167&172&177\\
\hline
$x$ &28&29&30&31&32&33&34&35&36&37&38&39&40&41 \\ \hline
$L(x)$ &103&108&113&118&122&125&130&135&140&145&149&153&158&163\\
$U(x)$ &182&188&194&199&204&209&214&219&224&229&234&239&244&250\\
\hline
$x$&42&43&44&45&46&47&48&49&50&51&52&53&54&55 \\ \hline
$L(x)$ &168&173&178&183&187&191&195&200&205&210&215&220&225&230\\
$U(x)$ &255&260&265&270&275&280&285&290&295&300&305&309&313&317\\
\hline
$x$&56&57&58&59&60&61&62&63&64&65&66&67&68&69 \\ \hline
$L(x)$ &235&240&245&250&256&261&266&271&276&281&286&291&296&301\\
$U(x)$ &322&327&332&337&342&347&351&355&360&365&370&375&378&382\\
\hline
$x$&70&71&72&73&74&75&76&77&78&79&80&81&82&83 \\ \hline
$L(x)$ &306&312&318&323&328&333&338&343&348&356&361&366&371&376\\
$U(x)$ &387&392&397&400&405&410&415&418&422&427&431&435&440&444\\
\hline
$x$&84&85&86&87&88&89&90&91&92&93&94&95&96&97 \\ \hline
$L(x)$ &383&388&393&398&406&411&416&423&428&436&441&448&454&461\\
$U(x)$ &447&453&455&460&463&468&471&475&478&484&485&488&492&495\\
\hline
$x$&98&99&100  \\ \hline
$L(x)$ &469&476&486 \\
$U(x)$ &497&499&500 \\
\hline
 \multicolumn{15}{c}{ Computational time: 0.0019 min. Total size: 7129}\\
\hline
\end{tabular}
\label{table:us.ne.Wang.500.100}
\end{table}

\begin{table}[!ht]
\caption{Confidence intervals given by W.~Wang's~\citeyearpar{Wang15} method~$\mC$  for $N = 500$, $n = 100$, and $\alpha = 0.05$.}
\begin{tabular}{l|llllllllllllll}
\hline
\hline
$x$  & 0 & 1 & 2  & 3  & 4  & 5  & 6  & 7 & 8 & 9 & 10 & 11  & 12 & 13  \\ \hline
$L(x)$ &0&1&3&5&8&12&15&17&22&25&29&33&37&40\\
$U(x)$ &16&24&32&39&46&52&59&65&72&78&84&90&95&101\\
\hline
$x$  & 14& 15 &16&17&18&19&20&21&22&23&24&25&26&27 \\ \hline
$L(x)$ &45&47&53&56&60&66&69&73&79&82&85&91&96&100\\
$U(x)$ &107&113&117&124&130&135&141&144&152&157&163&168&173&178\\
\hline
$x$ &28&29&30&31&32&33&34&35&36&37&38&39&40&41 \\ \hline
$L(x)$ &102&108&114&118&122&125&131&136&142&145&149&153&158&164\\
$U(x)$ &182&188&194&200&205&210&215&220&225&231&236&241&246&250\\
\hline
$x$&42&43&44&45&46&47&48&49&50&51&52&53&54&55 \\ \hline
$L(x)$ &169&174&179&183&187&191&195&201&206&211&216&221&226&232\\
$U(x)$ &253&258&263&268&274&279&284&289&294&299&305&309&313&317\\
\hline
$x$&56&57&58&59&60&61&62&63&64&65&66&67&68&69 \\ \hline
$L(x)$ &237&242&247&250&254&259&264&269&275&280&285&290&295&300\\
$U(x)$ &321&326&331&336&342&347&351&355&358&364&369&375&378&382\\
\hline
$x$&70&71&72&73&74&75&76&77&78&79&80&81&82&83 \\ \hline
$L(x)$ &306&312&318&322&327&332&337&343&348&356&359&365&370&376\\
$U(x)$ &386&392&398&400&404&409&415&418&421&427&431&434&440&444\\
\hline
$x$&84&85&86&87&88&89&90&91&92&93&94&95&96&97 \\ \hline
$L(x)$ &383&387&393&399&405&410&416&422&428&435&441&448&454&461\\
$U(x)$ &447&453&455&460&463&467&471&475&478&483&485&488&492&495\\
\hline
$x$&98&99&100  \\ \hline
$L(x)$ &468&476&484 \\
$U(x)$ &497&499&500 \\
\hline
 \multicolumn{15}{c}{ Computational time: 10.1792 min. Total size: 7131}\\
\hline
\end{tabular}
\label{table:Wang.ne.us.500.100}
\end{table}

In addition to the total sizes, Tables~\ref{table:us.ne.Wang.500.100} and \ref{table:Wang.ne.us.500.100} also show the computational times\footnote{Computed using R's \texttt{proc.time()} function} used by both methods, at the bottom of each table. Whereas $\mC^*$ took roughly $1/10$th of a  second ($.0019$ minutes) to fill the table, $\mC_W$ took more than 10 minutes. As mentioned above, this is due to the adjusting technique of $\mC_W$ which requires repeated updating of intervals, whereas $\mC^*$ just requires one pass through the acceptance intervals for adjustment. Figure~\ref{time.wang.us} gives a more complete comparison of computational times in this setting. The additional time required by $\mC_W$ is sizable, even exceeding $25$ minutes for values of $n$ near the middle of the range. A comparison of computational times of $\mC_{Piv}$ and $\mC^*$ is shown in Figure~\ref{time.pivot.us}, which shows that the times are much faster overall compared to $\mC_W$ (the longest times being less than $1/3$ of a second), and comparable between the two methods.

\begin{figure}[!ht]
    \centering
    \caption{The computational time of the confidence intervals $\mC_W$ and $\mC^*$ for $N = 500$, $\alpha = 0.05$, and $n=10, 20, \ldots,490$.} \label{time.wang.us}
\includegraphics[width=12cm]{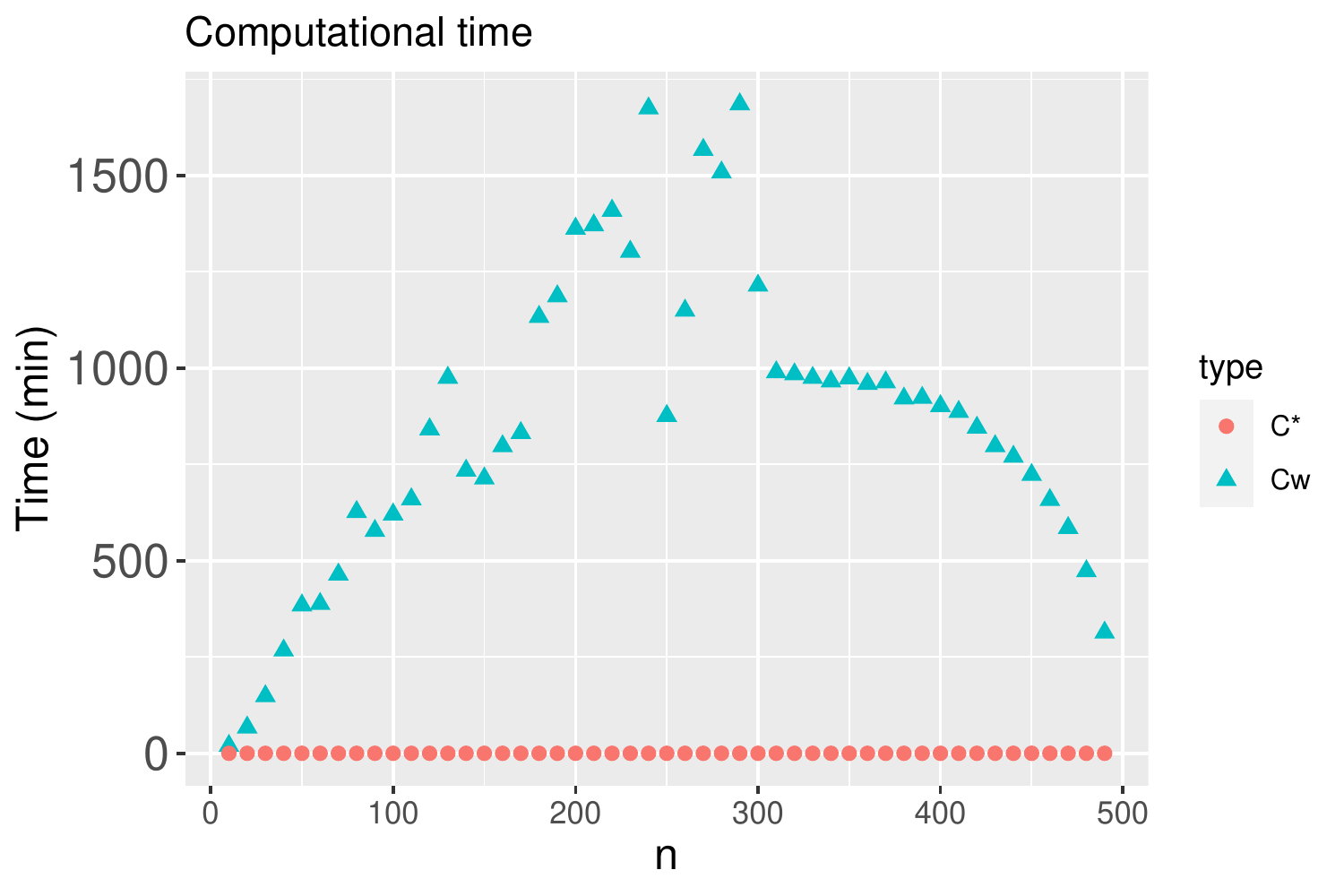} 
\end{figure}

\begin{figure}[!ht]
    \centering
    \caption{The computational time of the confidence intervals $\mC_{Piv}$ and $\mC^*$ for $N = 500$, $\alpha = 0.05$, and $n=10, 20, \ldots,490$.} \label{time.pivot.us}
\includegraphics[width=12cm]{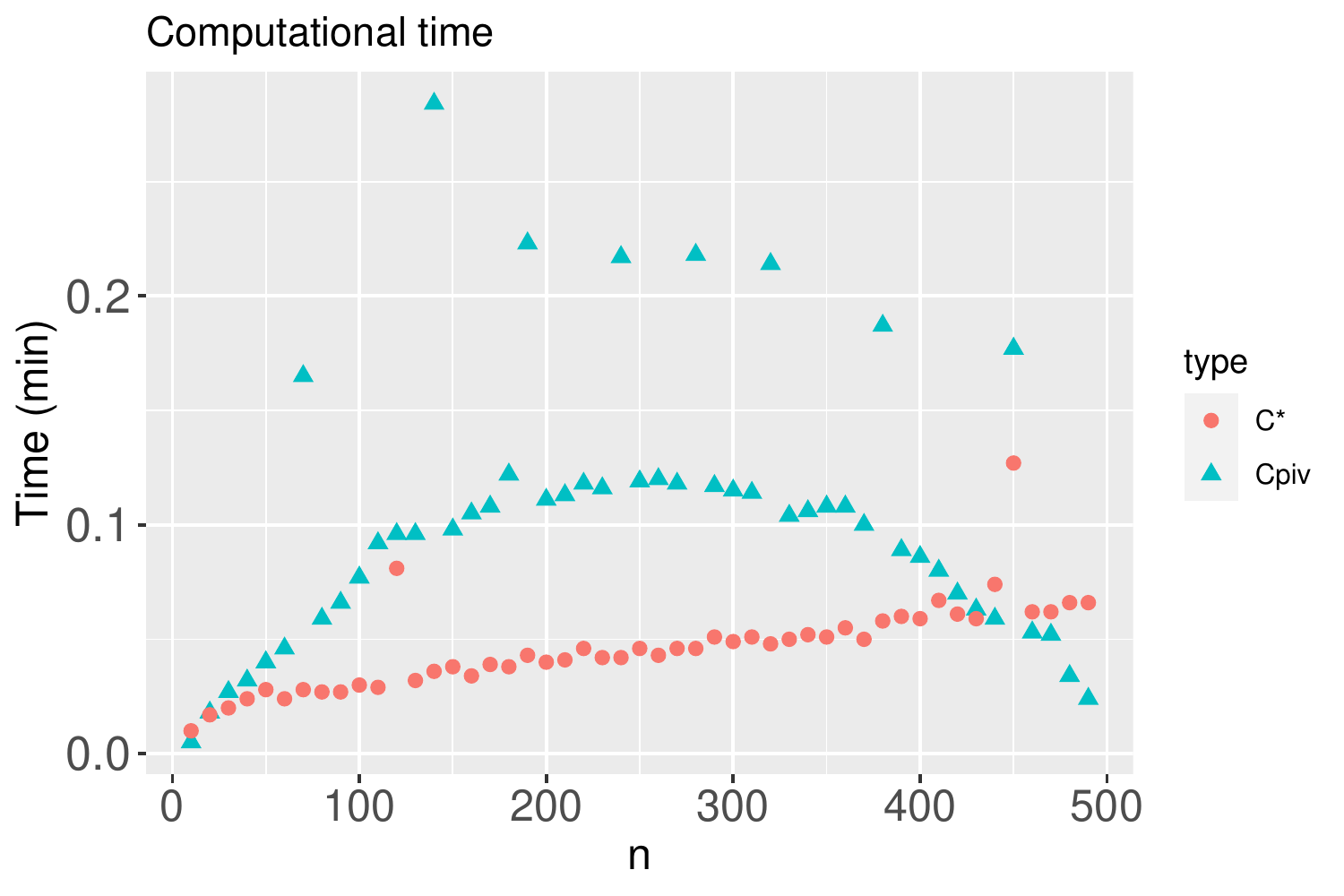} 
\end{figure}

Figures~\ref{cover.us.500.100.05}~and~\ref{cover.Wang.500.100.05} show the coverage probability for the $n=100$ case of the three methods as a function of $M=0,1,\ldots 500$.  Like their sizes, $\mC_W$ and $\mC^*$ have very similar coverage probabilities, whereas that of $\mC_{Piv}$  is overall higher (an \emph{undesirable} property once it exceeds $1-\alpha$), especially for values of $M$ near the endpoints~$0$ and $N=500$.

\begin{figure}[!ht]
    \centering
    \caption{Coverage probability of $\mC^*$  and $\mC_{Piv}$ for $N = 500$, $n = 100$, and $\alpha = 0.05$.} \label{cover.us.500.100.05}
\includegraphics[width=15cm]{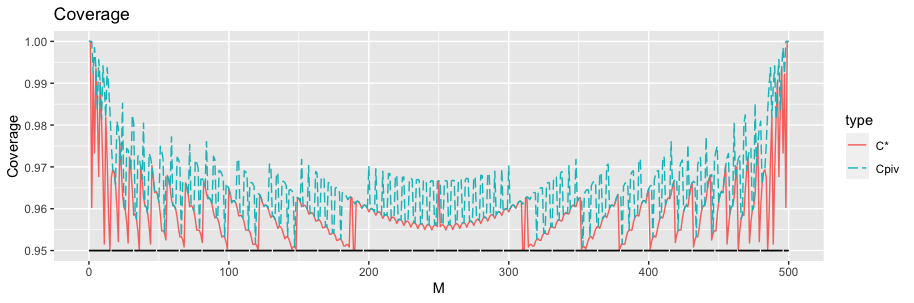} 
\end{figure}

\begin{figure}[!ht]
    \centering
    \caption{Coverage probability of $\mC_W$ for $N = 500$, $n = 100$, and $\alpha = 0.05$.} \label{cover.Wang.500.100.05}
\includegraphics[width=12cm, height=6cm]{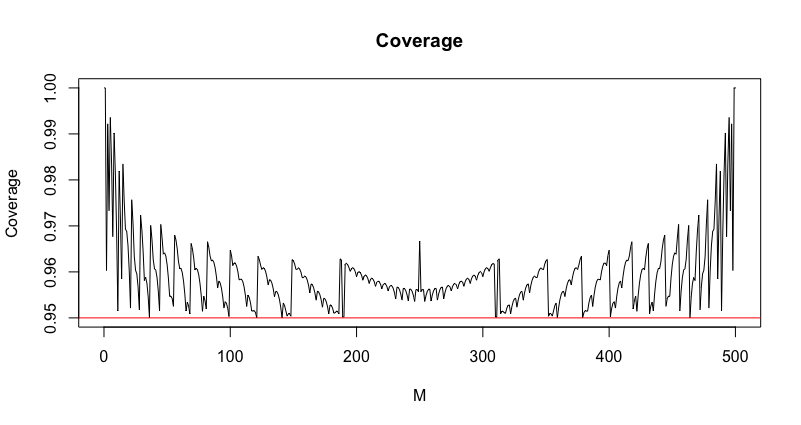}
\end{figure}
\end{ex}

\begin{ex}\label{ex:CW.time} We compare the computational time of $\mC^*$ and $\mC_W$ in the setting $\alpha = 0.05$, $N=200, 400, \ldots, 1000$, and $n=N/2$.  

Wang's~\citeyearpar{Wang15} method is time-consuming especially when the sample size and the population size are large. A comparison of computational times of $\mC^*$ and $\mC_W$ are shown in Figure~\ref{time.wang.us.1000}.  When $N = 1000$, $n = 500 $ and $\alpha = 0.05$, the computational time for $\mC_W$ reached 250 minutes, and $\mC^*$ only took 0.0111 minutes, which shows that our method is efficient.

\begin{figure}[!ht]
    \centering
    \caption{The computational time of the confidence intervals $\mC_W$ and $\mC^*$ for $N=200, 400, \ldots, 1000$, and $n=N/2$.} \label{time.wang.us.1000}
\includegraphics[width=12cm]{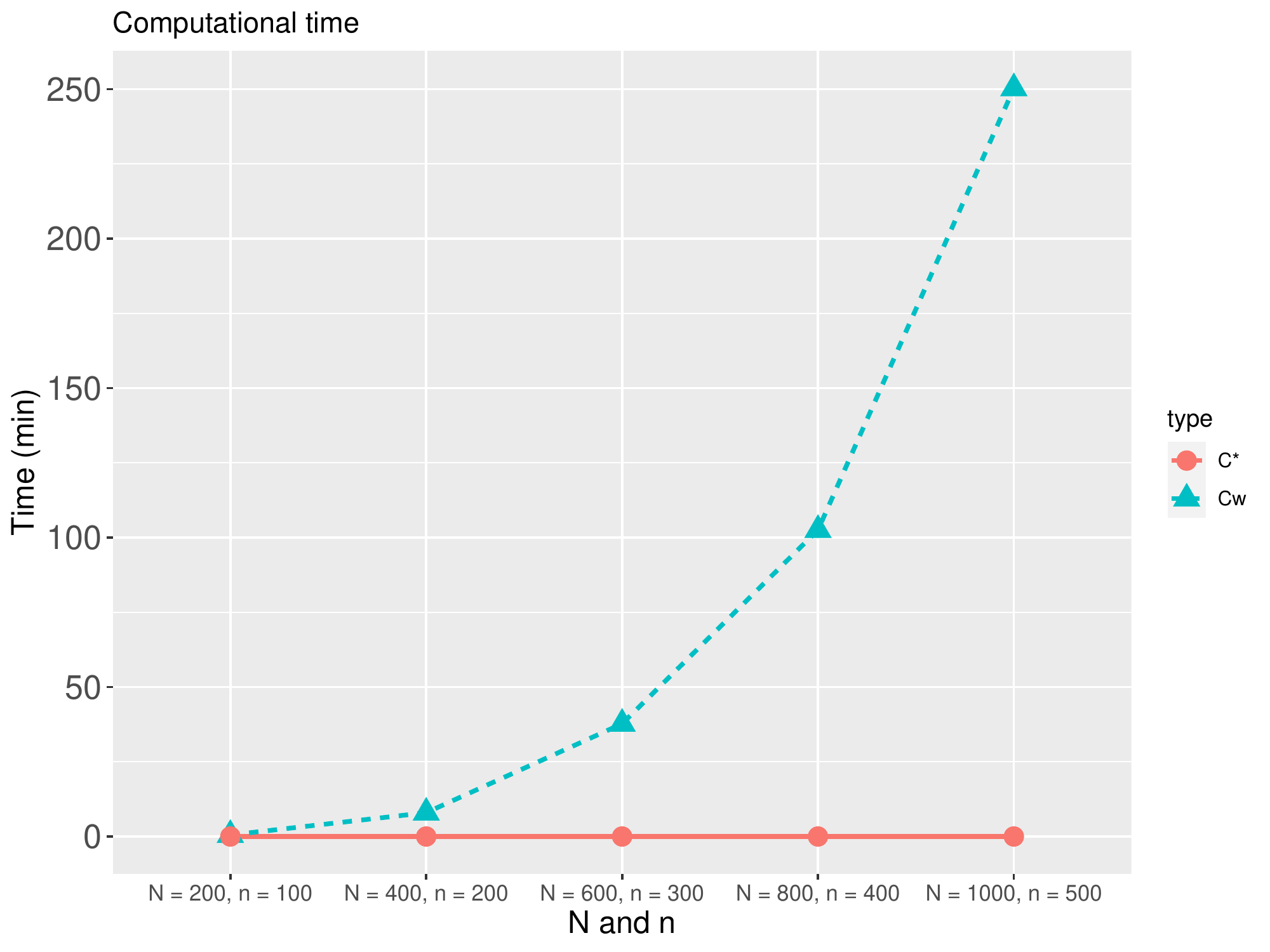} 
\end{figure}

\end{ex}

\begin{ex}\label{ex:AN/2.even} In this example we show the necessity of part~\ref{part:set.nonopt} of Theorem~\ref{thm:len.opt.set}. That is, we exhibit a setting with $n$, $N$, and $\mA(N/2)$ all even  for a certain acceptance set~$\mA$ whose inversion~$\mC$ is size-optimal with $|\mC|=|\mC^*|-1$. Set  $N = 20$, $n = 6$, and $\alpha = 0.6$. For $M\ne N/2=10$ define $\mA(M)=[a_M^*,b_M^*]$ to be the same intervals given by Theorem~\ref{thm:a*b*} and inverted to create $\mC^*$, and define $\mA(10)=\{2,4\}$. For all $M\ne 10$, $\mA(M)$ is a level-$\alpha$ interval, and $\mA(10)$ is as well since
\begin{equation*}
P_{M=10}(2)=P_{M=10}(4)=.244
\end{equation*} to $3$ decimal places, thus
$P_{M=10}(\{2,4\}) > .4 = 1-\alpha.$ It can be shown that $\mA^*(10)=[2,4]$, thus the intervals $\mA$ have $1$ fewer point than $\mA^*$, so by \eqref{lenCA=lenA} we have that $|\mC| = |\mC^*|-1$.
\end{ex}

\begin{ex}[Air quality data] In this example we apply our confidence interval~$\mC^*$ to data collected by China's Ministry of Environmental Protection (MEP) and discussed by \citet{Liang16}. The MEP collects data on particulate matter (PM$_{2.5}$) concentration, measured in $\mu g/m^3$, of fine inhalable particles with diameters less than 2.5 micrometers. The U.S.\ Environmental Protection Agency \citeyearpar{Environmental12} classifies the air quality of a given day as ``hazardous'' if the day's 24-hour average PM$_{2.5}$ measurement  exceeds the set threshold~$250.5$. \citet{Liang16} analyzed the 2013 to 2015 MEP data and concluded that it was consistent with measurements taken at nearby U.S.\ diplomatic posts, the U.S.\ Embassy in Beijing and four U.S\ Consulates in other cities. However, a persistent problem with the MEP data is a high degree of missing days.  For a given year, if the missing days are assumed to be missing at random with each day of the year equally likely, then the number~$X$ of remaining `hazardous' days, conditioned on the number~$n$ of remaining days, follows a hypergeometric distribution with $N=365$ and unknown actual number~$M$ of annual hazardous days, to be estimated as an indication of annual air quality.

We focus on the 2015 data from 3 MEP sites in Beijing: Dongsi, Dongsihuan, and Nongzhanguan. For each of these sites, Table~\ref{table:air} shows the number~$n$ of days with complete measurements, the observed number~$x$ of  days with complete measurements classified as hazardous, the point estimate~$Nx/n$ (with $N=365$) of the number~$M$ of annual hazardous days, and the 90\% confidence interval~$\mC^*(x)$ for $M$, which are also plotted in Figure~\ref{fig:air_MEP_example}.  The point estimates from the MEP sites are similar to and surround the estimate at the U.S.\ Embassy data, similar to the conclusions drawn by \citet{Liang16}. But the confidence intervals also show that the MEP estimates are more variable, largely in the direction of indicating worse air quality, with two out of three upper confidence limits being much larger for the MEP sites than for the U.S.\ Embassy.

\begin{table}[!ht]
\caption{For the Beijing air quality data \citep{Liang16}, the number~$n$ of days with complete measurements, the number~$x$ of  days with complete measurements classified as hazardous, the point estimate~$Nx/n$ (to 1 decimal place) of the number~$M$ of annual hazardous days, and the 90\% confidence interval~$\mC^*(x)$ for $M$.}
\begin{center}
\begin{tabular}{l|c|c|c|c}
Site&$n$&$x$&$Nx/n$&90\% CI for $M$\\ \hline
Dongsi&292&16&20.0 & $[17,24]$\\
Dongsihuan&166&7&15.4& $[10,24]$\\
Nongzhanguan&290&11&13.8& $[11,17]$\\
U.S.\ Embassy&332&15&16.5& $[15,18]$
\end{tabular}
\end{center}
\label{table:air}
\end{table}%

\begin{figure}[!ht]
    \centering
    \caption{The 90\% confidence interval~$\mC^*(x)$ for the number~$M$ of annual hazardous days at different locations in Beijing.} \label{fig:air_MEP_example}
\includegraphics[width=12cm]{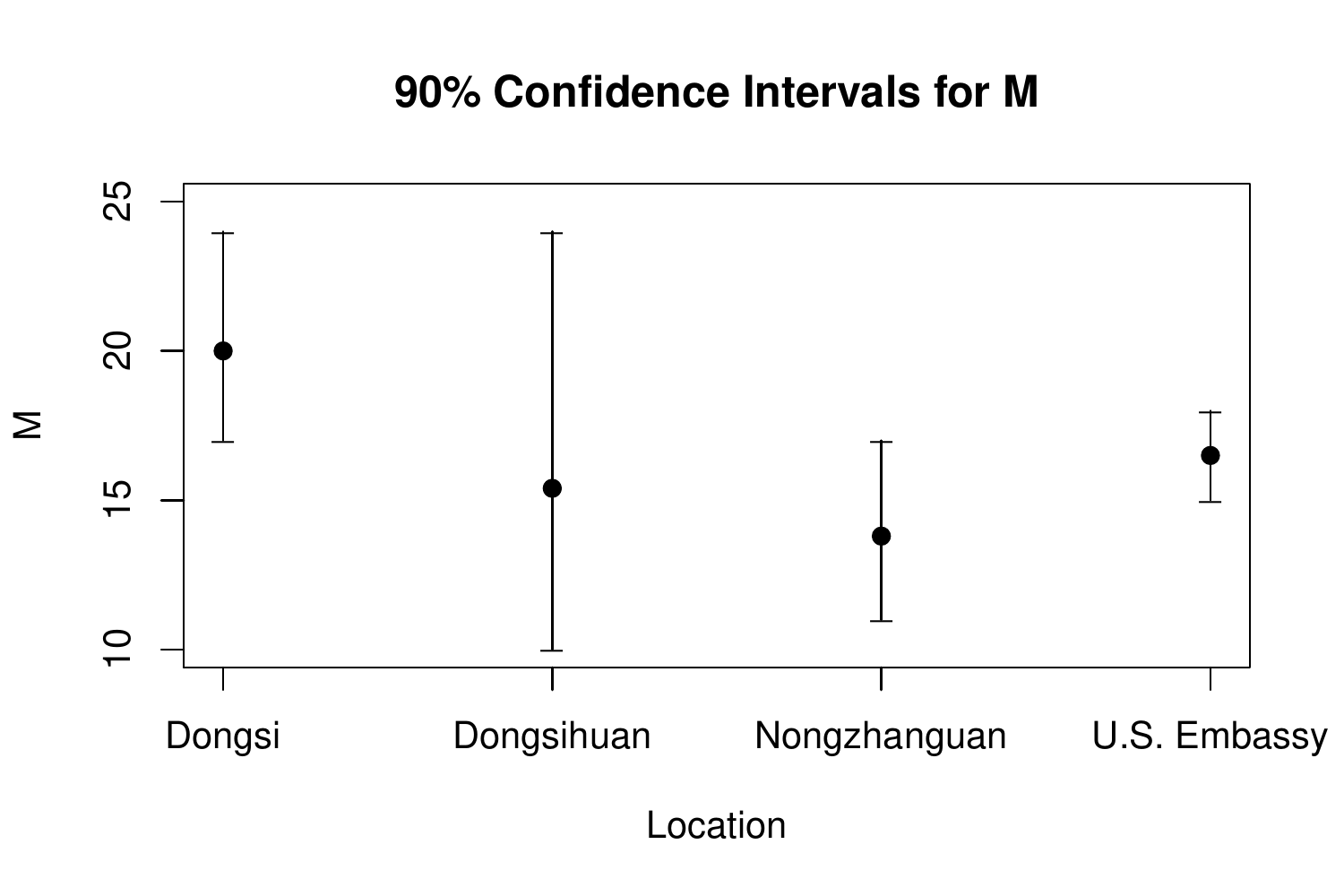} 
\end{figure}

\bigskip

%\citet{Liang16} impute the missing values by moving kernel weighted averages with 3hr bandwidth. They compute the percentage of time (hour) for $PM_{2.5} > 150$, $PM_{2.5} \leq 35$ and $PM_{2.5} > 35$. Form their results, we can see the percentage of time (hour) for $PM_{2.5} > 150$ from the MEP sites in Guangzhou are all $0\%$ from April 2013 to March 2015. 
%
%
%Guangzhou has two MEP sites, 5th Middle School and City Station. The City Station site actually has 243 hours with $PM_{2.5} > 150$. The site has more than $90\%$ missing data. After the imputation, they compute the percentage of time  as $0\%$. In this situation, our method could provide more information. This site reports 225 days' $PM_{2.5}$ completely from April 2013 to March 2015. 36 days have $PM_{2.5}$ reports exceeding 150. Our method provides the $95\%$ confidence interval $[90\,,149]$ as a guide for the number of days with $PM_{2.5}$ exceeding 150 from April 2013 to March 2015 (totally 730 days).
%
%
%
%
%The 5th Middle School site actually has 169 hours with $PM_{2.5} > 150$. The site has more than $70\%$ missing data. After the imputation, they compute the percentage of time (hour) for $PM_{2.5} > 150$ as $0\%$. This site reports 90 days' $PM_{2.5}$ completely from April 2013 to March 2015. 21 days have $PM_{2.5}$ reports exceeding 150. Our method provides the $95\%$ confidence interval $[117\,,238 ]$.

\end{ex}

\section{Discussion}
We have presented an efficient method of computing exact hypergeometric confidence intervals. Compared to the standard pivotal method, our method requires similar computational time but produces much shorter intervals. Our method produces intervals with total size no larger than, and strictly smaller than in some cases, the existing nearly-optimal method of W.~Wang~\citeyearpar{Wang15}, which is computationally much more costly than our and the pivotal method. Therefore we hope our method can provide something near the ``best of both worlds'' for this problem in terms of computational time and interval size.

The key to our method is the novel shifting of acceptance intervals before inversion, developed in Sections~\ref{sec: adj} and \ref{sec:mod.minacc}. We have observed in the numerical examples included in Section~\ref{sec:ex.comp}, as well as extensive further computations not included in this paper, that the needed shifts  in Theorem~\ref{thm:adj} seem to never exceed a single point. This is not needed in our theory but we  close by mentioning it as a tantalizing conjecture.

A similar approach to the one here of shifting optimal acceptance regions before inverting can be used to produce optimal confidence intervals for the hypergeometric population size~$N$ when it is unknown, such as in capture-recapture problems \citep{Bailey51,Pollock90,Wittes97}.  A forthcoming work will cover this problem.

%\bibliographystyle{apalike}
%\bibliography{../Bib_files/bibliography}

\begin{thebibliography}{}

\bibitem[Bailey, 1951]{Bailey51}
Bailey, N.~T. (1951).
\newblock On estimating the size of mobile populations from recapture data.
\newblock {\em Biometrika}, 38:293--306.

\bibitem[Blaker, 2000]{Blaker00}
Blaker, H. (2000).
\newblock Confidence curves and improved exact confidence intervals for
  discrete distributions.
\newblock {\em Canadian Journal of Statistics}, 28(4):783--798.

\bibitem[Blaker, 2001]{Blaker01}
Blaker, H. (2001).
\newblock Corrigenda: {C}onfidence curves and improved exact confidence
  intervals for discrete distribution.
\newblock {\em Canadian Journal of Statistics}, 29(4):681--681.

\bibitem[Blyth and Still, 1983]{Blyth83}
Blyth, C.~R. and Still, H.~A. (1983).
\newblock Binomial confidence intervals.
\newblock {\em Journal of the American Statistical Association},
  78(381):108--116.

\bibitem[Buonaccorsi, 1987]{Buonaccorsi87}
Buonaccorsi, J.~P. (1987).
\newblock A note on confidence intervals for proportions in finite populations.
\newblock {\em The American Statistician}, 41(3):215--218.

\bibitem[Casella and Berger, 2002]{Casella02}
Casella, G. and Berger, R.~L. (2002).
\newblock {\em Statistical Inference}.
\newblock Duxbury Press.

\bibitem[Chv{\'a}tal, 1979]{Chvatal79}
Chv{\'a}tal, V. (1979).
\newblock The tail of the hypergeometric distribution.
\newblock {\em Discrete Mathematics}, 25(3):285--287.

\bibitem[Clopper and Pearson, 1934]{Clopper34}
Clopper, C.~J. and Pearson, E.~S. (1934).
\newblock The use of confidence or fiducial limits illustrated in the case of
  the binomial.
\newblock {\em Biometrika}, 26(4):404--413.

\bibitem[Crow, 1956]{Crow56}
Crow, E.~L. (1956).
\newblock Confidence intervals for a proportion.
\newblock {\em Biometrika}, 43:423--435.

\bibitem[Crow and Gardner, 1959]{Crow59}
Crow, E.~L. and Gardner, R.~S. (1959).
\newblock Confidence intervals for the expectation of a {P}oisson variable.
\newblock {\em Biometrika}, 46(3/4):441--453.

\bibitem[Hald, 1990]{Hald90}
Hald, A. (1990).
\newblock {\em A History of Probability and Statistics and Their Applications
  Before 1750}.
\newblock Wiley, New York.

\bibitem[Johnson et~al., 1993]{Johnson93}
Johnson, N.~L., Kotz, S., and Kemp, A.~W. (1993).
\newblock {\em Univariate Discrete Distributions}.
\newblock John Wiley \& Sons, New York, second edition.

\bibitem[Keilson and Gerber, 1971]{Keilson71}
Keilson, J. and Gerber, H. (1971).
\newblock Some results for discrete unimodality.
\newblock {\em Journal of the American Statistical Association},
  66(334):386--389.

\bibitem[Konijn, 1973]{Konijn73}
Konijn, H.~S. (1973).
\newblock {\em Statistical Theory of Sample Survey Design and Analysis}.
\newblock North-Holland Publishing Company.

\bibitem[Liang et~al., 2016]{Liang16}
Liang, X., Li, S., Zhang, S., Huang, H., and Chen, S.~X. (2016).
\newblock {PM}$_{2.5}$ data reliability, consistency, and air quality
  assessment in five {C}hinese cities.
\newblock {\em Journal of Geophysical Research: Atmospheres}, 121:10--220.
\newblock Data sets available at
  \url{https://archive.ics.uci.edu/ml/datasets/PM2.5+Data+of+Five+Chinese+Cities}.

\bibitem[Pollock et~al., 1990]{Pollock90}
Pollock, K.~H., Nichols, J.~D., Brownie, C., and Hines, J.~E. (1990).
\newblock Statistical inference for capture-recapture experiments.
\newblock {\em Wildlife Society Monographs}, 107:3--97.

\bibitem[Rice, 2007]{Rice07}
Rice, J.~A. (2007).
\newblock {\em Mathematical Statistics and Data Analysis}.
\newblock Duxbury Press, Belmont, California, 3rd edition.

\bibitem[Schilling and Doi, 2014]{Schilling14}
Schilling, M.~F. and Doi, J.~A. (2014).
\newblock A coverage probability approach to finding an optimal binomial
  confidence procedure.
\newblock {\em The American Statistician}, 68(3):133--145.

\bibitem[Skala, 2013]{Skala13}
Skala, M. (2013).
\newblock Hypergeometric tail inequalities: {E}nding the insanity.
\newblock {\em arXiv preprint arXiv:1311.5939}.

\bibitem[Sterne, 1954]{Sterne54}
Sterne, T.~E. (1954).
\newblock Some remarks on confidence or fiducial limits.
\newblock {\em Biometrika}, 41:275--278.

\bibitem[{U.S.\ Environmental Protection Agency}, 2012]{Environmental12}
{U.S.\ Environmental Protection Agency} (2012).
\newblock {The National Ambient Air Quality Standards for Particle Pollution
  Revised Air Quality Standards For Particle Pollution And Updates To The Air
  Quality Index (AQI)}.
\newblock
  \url{https://www.epa.gov/sites/production/files/2016-04/documents/2012_aqi_factsheet.pdf}.

\bibitem[Wang, 2015]{Wang15}
Wang, W. (2015).
\newblock Exact optimal confidence intervals for hypergeometric parameters.
\newblock {\em Journal of the American Statistical Association},
  110(512):1491--1499.

\bibitem[Wittes, 1972]{Wittes97}
Wittes, J.~T. (1972).
\newblock Note: On the bias and estimated variance of {C}hapman's two-sample
  capture-recapture population estimate.
\newblock {\em Biometrics}, 28:592--597.

\end{thebibliography}

\def\cprime{$'$}

\appendix

\section{Properties of the hypergeometric distribution and auxiliary lemmas}\label{sec: general}

In this section, we first record some well known properties of the $\mbox{Hyper}(M,n,N)$ distribution in Lemmas~\ref{lem:hyp.basics} and \ref{lem:hyp.uni.x},  the latter covering unimodality of $P_M(x)$ in $x$.  The content of Lemma~\ref{lem:hyp.basics} is  mentioned in \citet[][Chapter~6]{Johnson93}, so we do not prove that here, and Lemma~\ref{lem:hyp.uni.x} follows from the expression
\begin{equation*}
\frac{P_M(x)}{P_M(x-1)} = \frac{(M+1-x)(n+1-x)}{x(N-M-n+x)}\qmq{for}1\le x\le n.
\end{equation*}
After that we state and prove some needed auxiliary results concerning other types of  monotonicity and unimodality: monotonicity of density function ratios with respect to $M$ in Lemma~\ref{prob ratio}, and unimodality of $P_M([a,b])$ with respect to $M$ in Lemma~\ref{lem:uni.intvl} and with respect to shifts in the interval~$[a,b]$ in Lemma~\ref{peak property}.  Throughout let $P_M(x)$ denote the density~\eqref{hyp.pdf} of the $\mbox{Hyper}(M,n,N)$ distribution.

\begin{lemma}\label{lem:hyp.basics} 
\begin{enumerate}
\item\label{bscs.lem.symms} We have 
\begin{align}
X\sim\mbox{Hyper}(n,M,N)&\Leftrightarrow X\sim\mbox{Hyper}(M,n,N)\nonumber\\
&\Leftrightarrow n-X\sim\mbox{Hyper}(N-M,n,N).\label{X.n-X.sym}
\end{align}
\item\label{bscs.lem.cplng} A useful coupling: For $n<N$, $X\sim\mbox{Hyper}(M,n+1,N)$ can be written  $X=X'+Y$ where $X'\sim\mbox{Hyper}(M,n,N)$ and $Y|X'\sim \mbox{Bern}((M-X')/(N-n))$.
\item Monotone likelihood ratio:  For every $M_1,M_2\in[N]$ with $M_1<M_2$, $P_{M_2}(x)/P_{M_1}(x)$ is nondecreasing in $x$ (with the convention $c/0=\infty$ for $c>0$).
\end{enumerate}
\end{lemma}

\begin{lemma}[Unimodality properties of the hypergeometric]\label{lem:hyp.uni.x} Let
\begin{equation}\label{argmax.x}
m = \frac{(n+1)(M+1)}{N +2},
\end{equation}
$m_1=\lceil m-1\rceil$, and $m_2=\lfloor m\rfloor$. The following hold.
\begin{enumerate}
\item\label{part:unimod.pos.PM} $P_M(x)$ increases strictly on $[x_{\min}, m_1]$ and decreases strictly on $[m_2,x_{\max}]$, where $x_{\min}=\max\{0,M+n-N\}$ and $x_{\max}=\min\{n,M\}$ are the smallest and largest, respectively, of the $x$ values with positive $P_M$ probability.
\item $\arg\max_x P_M(x) = [m_1,m_2]$, where $[m_1,m_2]=[\lfloor m\rfloor, \lfloor m\rfloor]$, unless $m$ is an integer, in which case $[m_1,m_2]=[m-1,m]$.
\end{enumerate}
\end{lemma}

The next lemma establishes monotonicity in $M$ of ratios $P_M(x_2)/P_M(x_1)$.

\begin{lemma} \label{prob ratio}
For fixed $N$ and $n$, let $x_1$, $x_2$ be distinct integers in $[0,n]$ such that $0<x_2-x_1< N - n$.  Then \begin{equation}\label{rat.lem.M.rng}
 \frac{P_M(x_2)}{P_M(x_1)} <\frac{P_{M + 1}(x_2)}{P_{M + 1}(x_1)}\qmq{for}  x_2 \leq M \leq N - n + x_1.
\end{equation}
\end{lemma}

\begin{proof} We have
 \begin{align} 
    \frac{P_M(x_2)}{P_M(x_1)} &= \prod^{x_2 - 1}_{x = x_1} \frac{(M - x)(n - x)}{(N - M - (n - x) + 1)(x+1)} \label{eq:1}\\
     & < \prod^{x_2 - 1}_{x = x_1} \frac{(M + 1 - x)(n - x)}{(N - (M + 1) - (n - x) + 1)(x+1)}
     = \frac{P_{M + 1}(x_2)}{P_{M + 1}(x_1)}. \label{eq:2}
\end{align}
\end{proof}

The next lemma establishes the unimodality of probabilities $P_M([a,b])$ as a function of $M$. It is helpful to define coupled random variables $X$ and $Y$ as the numbers of red and white balls, respectively, in a simple random sample of $n$ from a box of $N$ balls in which $M$ balls are white, one is red, and the remaining $N-(M+1)$ balls are green. Then $X\sim\mbox{Hyper}(M,n,N)$ and $X+Y\sim\mbox{Hyper}(M+1,n,N)$. In the usual notation, $P_{M+1}(x) = P(X+Y=x)$ and $P_M(x)=P(X=x)$. Writing
\begin{multline*}
P_{M+1}(x)-P_M(x) = [P(X=x-1,Y=1)+P(X=x,Y=0)]\\
-[P(X=x,Y=1)+P(X=x,Y=0)]\\
=P(X=x-1,Y=1)-P(X=x,Y=1),
\end{multline*}
and summing over $x$ from $a$ to $b$ yields
\begin{equation}\label{Pab.diff}
P_{M+1}([a,b])-P_M([a,b]) = P(X=a-1,Y=1)-P(X=b,Y=1).
\end{equation}
Note that for $x$ such that $P(X=x)>0$, 
\begin{equation}\label{red.ball.cond}
P(X=x,Y=1)= P(X=x)\frac{n-x}{N-M}= P_M(x)\frac{n-x}{N-M}
\end{equation}
since $x$ white balls in the sample implies that $n-x$ of the $N-M$ colored (red or green) balls are in the sample, so that the red ball has conditional probability $(n-x)/(N-M)$ of being in the sample. Relation \eqref{red.ball.cond} is trivially true when $P(X=x)=0$, so it is true for all $x\in[n]$. Using  \eqref{Pab.diff} and  \eqref{red.ball.cond}, 
\begin{equation}\label{Pab.diff.PaPb}
(N-M)(P_{M+1}([a,b])-P_{M}([a,b])) = (n-(a-1))P_M(a-1)-(n-b)P_M(b).
\end{equation} This equation provides the basis for the following lemma.

\begin{lemma}\label{lem:uni.intvl}
Assume $0\le a\le b\le n$ and $b-a<n$.  Then $P_M([a,b])$ is nondecreasing for $M\le M(a,b)$ and nonincreasing for $M\ge M(a,b)$, where
\begin{equation}\label{peak}
M(a,b)=\begin{cases}
0&\mbox{if $a=0$}\\
N&\mbox{if $b=n$}\\
\min\{M:(n-(a-1))P_M(a-1)<(n-b)P_M(b)\}&\mbox{otherwise.}
\end{cases}
\end{equation}
\end{lemma}

 \begin{proof} By \eqref{Pab.diff.PaPb}, 
 \begin{equation}\label{sgn.Pab.diff}
\mbox{sgn}(P_{M+1}([a,b])-P_{M}([a,b])) = \mbox{sgn}((n-(a-1))P_M(a-1)-(n-b)P_M(b)).
\end{equation} If $a=0$, the first term on the right-hand side vanishes and $\{P_M([a,b])\}$ is therefore nonincreasing. Similarly if $b=n$, $\{P_M([a,b])\}$ is nondecreasing. It remains to consider only $1\le a\le b\le n-1$, and it suffices to show that if $M_1, M_2$ are such that
\begin{equation}\label{sgnPab.diff.LHS}
\mbox{sgn}(P_{M+1}([a,b])-P_{M}([a,b])) = \begin{cases}
+1&\mbox{if $M=M_1$}\\
-1&\mbox{if $M=M_2$,}
\end{cases} 
\end{equation}
then $M_1<M_2$. Since the coefficients $n-(a-1)$ and $n-b$ in \eqref{sgn.Pab.diff} are positive, \eqref{sgn.Pab.diff} and \eqref{sgnPab.diff.LHS} imply that $P_{M_1}(a-1)$ and $P_{M_2}(b)$ must be positive. Therefore, since by \eqref{x.range} $P_M(x)$ must be positive if and only if $x\le M\le x+N-n$, we have 
\begin{equation*}
M_1\in I_1:=[a-1,a-1+N-n]\qmq{and} M_2\in I_2:= [b,b+N-n].
\end{equation*}
The endpoints of $I_1$ are less than the corresponding endpoints of $I_2$, so that $M_1$ cannot be to the right of $I_2$, and if it is to the left, $M_1<M_2$ follows immediately. So assume that $M_1$ belongs to $I_2$ and similarly that $M_2$ belongs to $I_1$. Then $M_1$ and $M_2$ both belong to $I_1\cap I_2$, hence
\begin{equation*}
b\le M_1,M_2\le a-1+N-n.
\end{equation*}
Since $n-b$ is positive and $P_M(a-1)$ is positive on this interval, \eqref{sgn.Pab.diff} and \eqref{sgnPab.diff.LHS} imply that
\begin{equation*}
\mbox{sgn}\left(\frac{P_M(b)}{P_M(a-1)} - \frac{n-(a-1)}{n-b}\right) = \begin{cases}
-1&\mbox{if $M=M_1$}\\
+1&\mbox{if $M=M_2$,}
\end{cases} 
\end{equation*} and the property~\eqref{rat.lem.M.rng} implies that $M_1<M_2$.
 \end{proof}

\begin{lemma} \label{peak property}
For fixed $n$, $N$, $0 \leq a \leq b < n$, and positive integer $ d \leq n - b$, we have 
\begin{enumerate}[label=(\roman*)]
\item\label{lem.peak.i} $M(a,b) \leq M(a+d, a+d)$, and 
\item\label{lem.peak.ii} $P_{M}([a, b]) \leq  P_{M} ([a +d, b+d])$ for all $M \geq M(a+d, b+d)$. 
\end{enumerate}

\end{lemma}

\begin{proof} The lemma can be proved by induction on $d$ as it is straightforward to verify using the definition~\eqref{peak} of $M(a,b)$ and the inequality~\eqref{eq:2}. We omit the details. 
\end{proof}

\section{Proof of Theorem~\ref{thm:adj} and auxiliary results}\label{sec:a.max.opt.proofs}

\begin{proof}[Proof of Theorem~\ref{thm:adj}] We first prove that $\mM_a$ and $\mM_b$ are disjoint. Take $M \in \mM_a$, so that $a_M<\oa_M$, and we will show that $M\not\in\mM_b$, i.e., $b_M\le b_{M'}$ for all $M'>M$. Fix such an $M'\in\mM$, and we will consider 2 cases, comparing $a_{M'}$ and $\oa_M$. \underline{Case 1}: $a_{M'} < \oa_M$. In this case we have $a_{M'} < \oa_M \leq \oa_{M'}$. If equality holds in this last, then there exists $M'_\ell < M < M'$ such that $a_{M'_\ell} = \oa_{M} = \oa_{M'}$, so
$b_{M'} \geq b_M$ by part~\ref{laom.d>1.2} of Lemma~\ref{adjust one more}. Otherwise, $\oa_M < \oa_{M'}$, so it must be that a new maximum is achieved between $M$ and $M'$, i.e., $\oa_{M'} = a_{M'_\ell}$ for some $M < M'_\ell < M'$. Then $a_{M'_\ell} > \oa_M$, so  $b_{M} < b_{M'_\ell}$ by part~\ref{the.thm.a.bond} of Lemma~\ref{the thm}.  We also have $a_{M'_\ell} = \oa_{M'} > a_{M'} $, so $b_{M'_\ell} \leq b_{M'}$ by part~\ref{lemma.a.prob} of Lemma~\ref{a_b_one_problem}. Combining these inequalities gives $b_{M'} > b_M$. \underline{Case 2}: $a_{M'} \ge \oa_M$. We have $b_{M'} >  b_M$ by part~\ref{the.thm.a.bond} of Lemma~\ref{the thm}, implying that $b_M = \ub_M$, satisfying the claim. 

The proof of that $M \in \mM_b$ implies $M \notin \mM_a$ is similar.

For monotonicity of the endpoints it suffices to show that $b_M^{adj}\le b_{M+1}^{adj}$ for $M\le N-1$, since the corresponding result for the sequence $\{a_M^{adj}\}$ is proved similarly. For all $M\le N-1$, $b_{M+1}^{adj}\ge \ub_{M+1}\ge \ub_M$, and if $M\not\in\mM_a$ then $\ub_M=b_M^{adj}$. Hence only the case $M\in\mM_a$ remains to be considered. If also $M+1\in\mM_a$, then Lemma~\ref{adjust one more}, part~\ref{laom.d>1}, applies with $M^*=M+1$, hence $b_M-a_M\le b_{M+1}-a_{M+1}$, and since $\oa_M\le \oa_{M+1}$, the result follows by adding the last two inequalities. If $M+1\not\in \mM_a\cup \mM_b$, then $a_{M+1}=\oa_{M+1}\ge \oa_M$, and Lemma~\ref{the thm} applies, yielding $b_{M+1}\ge b_M+\oa_M-a_M$, which suffices. Finally, if $M+1\in\mM_b$, then $M+1<N$ and 
$$b_{M+1}>b_{M+1}^{adj}=\ub_{M+1}=b_{M'}\qmq{for some} M'>M+1.$$
By Lemma~\ref{a_b_one_problem}, $a_{M'}\ge a_{M+1}=\oa_{M+1}\ge \oa_M$, the equality holding since $M+1\not\in\mM_a$ by the disjointness of $\mM_a$ and $\mM_b$. Applying Lemma~\ref{the thm} and the definition of $M'$, $b_M^{adj}\le B_{M'}=\ub_{M+1}=b_{M+1}^{adj}$.

That the adjusted intervals are level-$\alpha$ is handled in parts~\ref{laom.lvl} and \ref{laom.lvl.Mb} of Lemma~\ref{adjust one more}, respectively, for the two nontrivial cases. Finally, note that the adjusted intervals have the same length as the original intervals, thus implying length optimality.
\end{proof}

The next  lemma establishes that anywhere a ``gap'' $a_M> \oa_M$ occurs in the sequence of lower endpoints of \AMO{} acceptance intervals, the gap may be ``filled'' by shifting the interval up the needed amount while maintaining the interval's acceptance probability and without violating monotonicity in the upper endpoint~$b_{M}$.

\begin{lemma} \label{adjust one more}
Let $\{[a_M, b_M]:\; M\in\mM\}$ be \AMO{} with $\mM\subseteq [N]$ an interval, and $\oa_M, \ub_M, \mM_a, \mM_b$ as defined  in \eqref{oa.ub.def}-\eqref{Ma.Mb.def}.
\begin{enumerate}
\item\label{laom.oa>a} If $M^* \in \mM_a$ then, letting $\delta = \oa_{M^*} - a_{M^*}$, we have
\begin{enumerate}
\item\label{laom.lvl}  $P_{M^*}\left( [a_{M^*} + \delta , b_{M^*} +\delta]\right) \geq 1 - \alpha$,
\item\label{laom.exist}  there exists $M_\ell \in \mM$ with $M_\ell<M^*$ such that $a_{M_\ell} = \oa_{M^*}$, 
\item\label{laom.b.shft} $b_{M^*} + \delta > b_{M_\ell} $ for any $M_\ell$ satisfying \ref{laom.exist},
\item\label{laom.d>1} for any $M_\ell$ satisfying \ref{laom.exist}, then for all $M\in[M_\ell, M^*]$ we have
\begin{enumerate}
\item $b_{M} - a_{M} \leq b_{M^*} - a_{M^*}$,
\item\label{laom.d>1.2} $b_{M} \leq b_{M^*}$.
\end{enumerate}
\end{enumerate}
\item\label{laom.ub<b}  If $M^* \in \mM_b$ then, letting $\delta = b_{M^*} - \ub_{M^*}$, we have
\begin{enumerate}
\item\label{laom.lvl.Mb} $P_{M^*}\left( [a_{M^*} - \delta , b_{M^*} - \delta]\right) \geq 1 - \alpha$,
\item\label{laom.exist2}  there exists $M_u \in \mM$ with $M_u>M^*$ such that $b_{M_u} = \ub_{M_u}$, 
 \item $a_{M^*} - \delta < a_{M_u} $ for any $M_u$ satisfying \ref{laom.exist2},
\item for any $M_u$  satisfying \ref{laom.exist2}, then for all $M^*\le M\le M_u$ we have
\begin{enumerate}
\item $b_{M} - a_{M} \leq b_{M^*} - a_{M^*}$,
\item $a_{M} \geq a_{M^*}$.
\end{enumerate}
\end{enumerate}
\end{enumerate}

\end{lemma}

 \begin{proof}[Proof of Lemma~\ref{adjust one more}] Part~\ref{laom.exist} follows from the definition of $\oa_M$. For part~\ref{laom.b.shft}, $b_{M^*} \geq b_{M_\ell}$ by Lemma \ref{a_b_one_problem}, so $b_{M^*} + \delta > b_{M_\ell}$.  Using  this and the fact that $a_{M^*} + \delta = \oa_{M^*} = a_{M_\ell}$, we have
\begin{equation*}
[a_{M_\ell} ,b_{M_\ell}] \subseteq [a_{M^*}+\Delta ,b_{M^*}+\Delta]\qmq{for all} \Delta\in[\delta],
\end{equation*} and thus
\begin{equation}\label{laom.PDelt}
P_M([a_{M_\ell} ,b_{M_\ell}]) \le P_M([a_{M^*}+\Delta ,b_{M^*}+\Delta])\qmq{for all} M\in\mM,\; \Delta\in[\delta].
\end{equation}  We know that 
 \begin{equation}\label{eq: adjust one more 1}
 P_{M^*} \left( [a_{M^*}+\delta ,b_{M^*} + \delta] \right) \leq P_{M^*} \left( [a_{M^*} ,b_{M^*}] \right) 
 \end{equation}
by Definition \ref{def:a.opt} since these intervals have the same width. If equality holds in \eqref{eq: adjust one more 1} then part~\ref{laom.lvl}  is proved because the right-hand-side is $\ge 1-\alpha$. Otherwise strict inequality holds in \eqref{eq: adjust one more 1}, which implies that $M^* < M(a_{M^*}+\delta ,b_{M^*} + \delta)$ by Lemma~\ref{peak property}, part~\ref{lem.peak.ii}. Then, using unimodality and \eqref{laom.PDelt}, we have 
\begin{align*}
 P_{M^*} \left( [a_{M^*}+\delta ,b_{M^*} + \delta] \right) & \geq P_{M_\ell} \left( [a_{M^*}+\delta ,b_{M^*} + \delta] \right)\\
    & \geq P_{M_\ell} \left( [a_{M_\ell} ,b_{M_\ell}] \right) \\
    & \geq 1 - \alpha,
\end{align*} finishing the proof of part~\ref{laom.lvl}.  For part~\ref{laom.d>1}, using unimodality we have
\begin{multline*}
P_{M^*}( [a_{M^*} + \delta , b_{M^*} +\delta]) \geq \\
  \min (   P_{M^*} \left( [a_{M^*}+\delta ,b_{M^*} + \delta] \right), P_{M_\ell} \left( [a_{M^*}+\delta ,b_{M^*} + \delta] \right)) \\
  \geq 1 - \alpha,
\end{multline*} and therefore
$$b_{M} - a_{M}  \leq b_{M^*} +\delta -  (a_{M^*} + \delta) = b_{M^*} - a_{M^*}$$ by length optimality of $[a_M,b_M]$.
By this inequality, if $a_{M^*} \geq a_{M}$ then $b_{M^*} \geq b_{M}$. Otherwise, $a_{M^*} < a_{M}$ so $b_{M^*} \geq b_{M}$ by Lemma \ref{a_b_one_problem}, completing the proof of \ref{laom.d>1}.

The proof of part~\ref{laom.ub<b} involves similar arguments, after reflecting the endpoint sequences $\wtilde{a}_{N - M} = n - b_M$, $\wtilde{b}_{N - M} = n - a_M$ and using Lemma~\ref{lemma: opt sym}. We omit the rest of the details.
\end{proof}

Parts~\ref{the.thm.a.bond}-\ref{the.thm.a.mono} of next lemma establish that the adjusted acceptance intervals given in Theorem~\ref{thm:adj}  have nondecreasing lower endpoints, and parts~\ref{the.thm.b.bond}-\ref{the.thm.b.mono} show the same for the upper endpoints.

\begin{lemma} \label{the thm}
Let $\{[a_M, b_M]:\; M\in\mM\}$ be \AMO{} with $\mM\subseteq [N]$ an interval, and $\oa_M, \ub_M, \mM_a, \mM_b$ as defined  in \eqref{oa.ub.def}-\eqref{Ma.Mb.def}.
\begin{enumerate}
\item\label{the.thm.a.bond} If $M^* \in \mM_a$ and $M^* < M \in \mM$ satisfy $a_{M}  \geq \oa_{M^*}$, then $b_{M^*}  < b_{M^*} + \oa_{M^*} - a_{M^*} \leq b_{M}$.
\item\label{the.thm.a.mono} The sequence $b_{M} + \oa_M - a_M$ is nondecreasing in $M \in \mM_a$.
\item\label{the.thm.b.bond} If $M^* \in \mM_b$ and $M^* > M \in \mM$ satisfy $b_{M}  \leq \ub_{M^*}$, then $a_{M^*} > a_{M^*} - (b_{M^*} - \ub_{M^*}) \geq a_{M}$.
\item\label{the.thm.b.mono} The sequence $a_{M} - b_{M} +\ub_{M}$ is nondecreasing in $M \in \mM_b$.
\end{enumerate}
\end{lemma}

\begin{proof} For part~\ref{the.thm.a.bond}, there must be $M_\ell< M^*$ such that $a_{M_\ell} = \oa_{M^*}>a_{M^*}$. Then we have $b_{M_\ell } \leq b_{M^*}$ by Lemma \ref{a_b_one_problem}, part~\ref{lemma.a.prob}. Combining these two, we have $[a_{M_\ell} , b_{M_\ell}] \subsetneqq [a_{M^*} , b_{M^*}]$, thus
$$
P_{M^*}( [a_{M_\ell} , b_{M_\ell}]) < 1- \alpha\le P_{M_\ell}( [a_{M_\ell} , b_{M_\ell}]),
$$
by length optimality of the latter. This implies that $M^*\geq M(a_{M_\ell} , b_{M_\ell})$ by Lemma~\ref{lem:uni.intvl}. We also have that
\begin{equation}\label{eq: the thm 1}
P_{M}( [a_{M_\ell} , b_{M_\ell}]) \leq P_{M^*}( [a_{M_\ell} , b_{M_\ell}]) < 1- \alpha    
\end{equation}
since $M > M^* \geq M(a_{M_\ell} , b_{M_\ell})$. If it were that $b_{M} \leq b_{M_\ell}$, then we would have $[a_{M} , b_{M}] \subseteq [a_{M_\ell} , b_{M_\ell}]$ and then \eqref{eq: the thm 1} would imply that
$$P_{M}( [a_{M} , b_{M}]) \leq P_{M}\left( [a_{M_\ell} , b_{M_\ell}]\right) < 1 - \alpha,$$ a contradiction. Thus it must be that $b_{M} > b_{M_\ell}$. Then we have $[a_{M_\ell} , b_{M_\ell}] \subseteq [a_{M_\ell} , b_{M}]$ and $[a_{M} , b_{M}] \subseteq [a_{M_\ell} , b_{M}]$,
$$P_{M_\ell}( [a_{M_\ell} , b_{M}]) \geq P_{M_\ell}( [a_{M_\ell} , b_{M_\ell}]) \geq 1 - \alpha$$ and $$P_{M}( [a_{M_\ell} , b_{M}]) \geq P_{M}( [a_{M} , b_{M}]) \geq 1 - \alpha\,.$$ Thus, by unimodality,
$P_{M^*}( [a_{M_\ell} , b_{M}])  \geq 1  - \alpha$, and so by length optimality we have
\begin{equation}\label{eq: the thm 2}
b_{M^*} - a_{M^*} \leq b_{M} - a_{M_\ell}.    
\end{equation}
Note that $b_{M^*}  + \oa_{M^*} - a_{M^*} = a_{M_\ell} + b_{M^*} - a_{M^*}$. By the inequality~\eqref{eq: the thm 2}, we have 
$$b_{M^*}  + \oa_{M^*} - a_{M^*}  = a_{M_\ell} + b_{M^*} - a_{M^*}  \leq a_{M_\ell} + b_{M} - a_{M_\ell} = b_{M},$$
concluding the proof of part~\ref{the.thm.a.bond}.

For \ref{the.thm.a.mono}, consider $M_1,M_2 \in \mM_a$ with $M_1 < M_2$. We have $\oa_{M_1} \leq \oa_{M_2}$. If strict inequality holds then there is $M_{\ell,2} \in \mM$ satisfying that $M_1 < M_{\ell,2} < M_2$ and $a_{M_{\ell,2}} = \oa_{M_2} > \oa_{M_1}$. Then using part~\ref{the.thm.a.bond} and Lemma~\ref{adjust one more}, respectively, for the following inequalities,
$$ b_{M_1} + \oa_{M_1} - a_{M_1} \leq b_{M_{\ell,2}} < b_{M_2} + \oa_{M_2} - a_{M_2}.$$
Otherwise $\oa_{M_1} = \oa_{M_2}$ whence there is $M_{\ell,1} \in \mM$ with $M_{\ell,1} < M_1$ and $a_{M_{\ell,1}} = \oa_{M_1}= \oa_{M_2}$, thus $b_{M_1} - a_{M_1} \leq b_{M_2}  - a_{M_2}$ by Lemma~\ref{adjust one more}. This establishes part~\ref{the.thm.a.mono}.

The proof of parts~\ref{the.thm.b.bond}-\ref{the.thm.b.mono}  involves similar arguments, after reflecting the endpoint sequences $\wtilde{a}_{N - M} = n - b_M$, $\wtilde{b}_{N - M} = n - a_M$ and using Lemma~\ref{lemma: opt sym}. We omit the rest of the details.
\end{proof}

The next lemma shows that \AMO{} intervals for $\mM$ can be reflected across $N/2$ to produce \AMO{} intervals in the reflected set for $\mM'=\{N-M:\; M
\in \mM\}$.

\begin{lemma}\label{lemma: opt sym}
If $\{[a_M, b_M]:\; M \in \mM\}$ are \AMO{} for $\mM$, then  $\{[\wtilde{a}_M := n - b_{N - M}, \wtilde{b}_M := n - a_{N - M}] :\; M \in \wtilde{\mM}\}$ are  \AMO{} for $\wtilde{\mM}:=\{N-M:\; M \in \mM\}$.
\end{lemma}

\begin{proof}[Proof of Lemma~\ref{lemma: opt sym}] Fix $M\in\mM$ and omit it from the notation. Let $X\sim\mbox{Hyper}(M,n,N)$ so that $\wtilde{X}:=n-X\sim\mbox{Hyper}(N-M,n,N)$ by \eqref{X.n-X.sym}. The three properties in Definition~\ref{def:a.opt} are straightforward to verify using that $P(\wtilde{X}\in[a,b]) = P(X\in[n-b,n-a])$. We omit the details.\end{proof}

The next lemma establishes the point masses inside a probability-maximizing (e.g., \AMO{}) set  have probabilities no less than those outside, and is used to prove that monotonicity can only be violated one endpoint at a time in \AMO{} acceptance intervals, and that probability maximizing sets must be intervals.

\begin{lemma} \label{larger} A set~$S\subseteq [x_{\min},x_{\max}]$ is $P_M$-maximizing if and only if $x\in S$, $y\not\in S$ implies that $P_M(x)\ge P_M(y)$.
\end{lemma}

\begin{proof}[Proof of Lemma~\ref{larger}] The condition is obviously necessary since otherwise replacing $x$ in $S$ by $y$ would increase $P_M(S)$. The converse follows by summing over $x\in S\setminus S^*$ vs.\ $y\in S^*\setminus S$ for $|S^*|=|S|$.
\end{proof}

 The following lemma establishes that, in \AMO{} intervals, monotonicity can only be violated one endpoint at a time. This property is used to prove that the sets~$\mM_a$ and $\mM_b$ in Theorem~\ref{thm:adj} are disjoint.

\begin{lemma} \label{a_b_one_problem} Suppose $[a,b]$ and $[a',b']$ are $P_M$ and $P_{M'}$ maximizing, respectively, and $M'>M$.
\begin{enumerate}[label=(\roman*)]
    \item\label{lemma.a.prob} If $a'<a$, then $b'\ge b$.
    \item\label{lemma.b.prob} If $b'<b$, then $a'\ge a$.
    \item\label{lemma.or.prob} $a'\ge a$ or $b'\ge b$.
\end{enumerate}
\end{lemma}

\begin{proof}[Proof of Lemma~\ref{a_b_one_problem}] It suffices to prove part~\ref{lemma.a.prob} since \ref{lemma.a.prob}, \ref{lemma.b.prob}, and \ref{lemma.or.prob} are logically equivalent.

Assume $a'<a$. Then $a'$ is outside $[a,b]$, and hence $P_M(a')\le P_M(b)$ by Lemma~\ref{larger}. Then $P_{M'}(a')<P_{M'}(b)$ by  Lemma~\ref{prob ratio}, the hypotheses of which follow from \eqref{x.range} and the fact that the intervals are probability maximizing, hence its points have positive probability. Since $[a',b']$ is $P_{M'}$ maximizing, it must contain $b$.
\end{proof}

\section{Auxiliary results for Section~\ref{sec:mod.minacc}}

\begin{lemma}\label{lem:a*b*.mono}
Under the assumptions of Theorem~\ref{thm:a*b*}, the intervals $\{[a^*_M, b^*_M]:\; M\in[N]\}$ defined by \eqref{a*b*.def} have nondecreasing endpoint sequences.
\end{lemma}

\begin{proof} Monotonicity holds separately for $M$ in the first and second cases of \eqref{a*b*.def} by Theorem~\ref{thm:adj}. We must show monotonicity ``across'' $N/2$, i.e., 
\begin{equation}\label{a*.mono.Neven}
a_{N/2-1}^*\le a_{N/2}^*\le a_{N/2+1}^*\qmq{and}b_{N/2-1}^*\le b_{N/2}^*\le b_{N/2+1}^*\qm{if $N$ is even, }
\end{equation}and
\begin{equation}\label{a*.mono.Nodd}
a_{\lfloor N/2\rfloor}^*\le a_{\lfloor N/2\rfloor+1}^*\qmq{and}b_{\lfloor N/2\rfloor}^*\le b_{\lfloor N/2\rfloor+1}^*\qm{if $N$ is odd.}
\end{equation}
For \eqref{a*.mono.Neven} take $N$ even, and both sets of inequalities in \eqref{a*.mono.Neven} are equivalent to
\begin{equation}\label{a*.mono.Neven.alt}
a_{N/2-1}^{adj}\le h_{\alpha/2}\le n	-b_{N/2-1}^{adj}.
\end{equation}
By Lemma~\ref{lem:a.largest} we have $a_{N/2} \geq h_{\alpha/2} \geq a_M$ for $M \in [N/2 - 1]$, so $a_{N/2} = \oa_{N/2}$, i.e., $N/2 \notin \mM_a$. Note that, by virtue of \eqref{oa.ub.def}, the upper endpoint~$b_M$ for the largest $M$ in the index set never gets adjusted. Here the index set is $[N/2]$ so $\ub_{N/2} =b_{N/2}$, thus $N/2 \notin \mM_b$. Therefore $[a_{N/2},b_{N/2}] = [a^{adj}_{N/2},b^{adj}_{N/2}]$. By  Lemma~\ref{lem:a.largest}  
 we also have that $b_{N/2}\le n-h_{\alpha/2}$, and combining these last two gives 
 $$h_{\alpha/2}\le n- b_{N/2} = n-b^{adj}_{N/2}\le n-b^{adj}_{N/2-1},$$  giving the second inequality in \eqref{a*.mono.Neven.alt}.
 
For the first inequality in \eqref{a*.mono.Neven.alt}, if $a_{N/2-1}\ge a_{N/2-1}^{adj}$ then $a_{N/2-1}^{adj}\le a_{N/2-1}\le h_{\alpha/2}$, using Lemma~\ref{lem:a.largest} for this last inequality. Otherwise, $a_{N/2-1}< a_{N/2-1}^{adj}$, meaning $N/2 - 1 \in \mM_a$ so, in particular, there exists $M^* < N/2 - 1$ such that $a_{M^*} = \oa_{N/2 - 1} = a^{adj}_{N/2 - 1}$.  Then $a^{adj}_{N/2 - 1}=a_{M^*}\le h_{\alpha/2}$, using Lemma~\ref{lem:a.largest} for the inequality.  

To prove \eqref{a*.mono.Nodd} take $N$ odd and let $M^*=\lfloor N/2\rfloor$. Both the inequalities in \eqref{a*.mono.Nodd} are equivalent to 
\begin{equation*}
a_{M^*}^{adj}+b_{M^*}^{adj}\le n.
\end{equation*}
If no adjustment is applied to $[a_{M^*}, b_{M^*}]$, i.e., $[a_{M^*}, b_{M^*}] = [a^{adj}_{M^*},b^{adj}_{M^*}]$, then this holds by Lemma~\ref{lemma: n - b > a}. Otherwise $M^*$ is in $\mM_a$ or $\mM_b$. If the latter then $[a^{adj}_{M^*},b^{adj}_{M^*}]$ is shifted down from $[a_{M^*}, b_{M^*}]$, i.e., $a^{adj}_{M^*}\le a_{M^*}$ and $b^{adj}_{M^*}\le b_{M^*}$, thus $$a_{M^*}^{adj}+b_{M^*}^{adj}\le a_{M^*}+b_{M^*}\le n$$ using Lemma~\ref{lemma: n - b > a}. 

Otherwise $M^* \in \mM_a$ meaning there is $M' < M^*$ such that $a_{M'} = \oa_{M^*} = a^{adj}_{M^*}$. By Lemma~\ref{lemma: n - b > a} we know that
\begin{equation}\label{M'.mdpnt.n}
a_{M'}+b_{M'}\le n.
\end{equation}
Define
\begin{equation*}
a'_M=\begin{cases}
a_M,&\mbox{for $M\in[M^*]$}\\
n-b_{N-M},&\mbox{for $M^*<M\le N$,}
\end{cases}
\quad
b'_M=\begin{cases}
b_M,&\mbox{for $M\in[M^*]$}\\
n-a_{N-M},&\mbox{for $M^*<M\le N$.}
\end{cases}
\end{equation*} Then $\{[a'_M,b'_M]:\; M\in[N]\}$ are \AMO{} since the $[a_M,b_M]$ are. We now apply part~\ref{the.thm.a.bond} of Lemma~\ref{the thm} to this set, with $M:=N-M'>M^*$. We have 
$$a'_M = n-b_{M'}\ge a_{M'} = \oa_{M^*} = \overline{a}'_{M^*},$$ where the inequality is by \eqref{M'.mdpnt.n}. Then using Lemma~\ref{the thm} for the following inequality, 
$$b_{M^*}^{adj} = b_{M^*}+\oa_{M^*}-a_{M^*}= b'_{M^*}+\overline{a}'_{M^*}-a'_{M^*}\le b_M'=n-a_{M'}=n-a_{M^*}^{adj},$$ the desired inequality.
\end{proof}

\begin{lemma} \label{lem:a.largest} If $N$ is even and $\{[a_M, b_M]:\; M\in[N/2]\}$ is any  \AMO{} set, then
\begin{enumerate}[label=(\roman*)]
\item\label{prt:cntr.in.hn-h} $[a_{N/2}, b_{N/2}] \subseteq [h_{\alpha/2}, n- h_{\alpha/2}]$,
\item\label{prt:a<h} for all   $M\in[N/2 - 1]$ we have $P_{N/2}( [a_M, n - a_M]) \geq 1- \alpha$ and $a_M\le h_{\alpha/2}$.
\end{enumerate}
\end{lemma}

\begin{proof} For part~\ref{prt:cntr.in.hn-h}, toward contradiction suppose that $a_{N/2} < h_{\alpha/2} $. Since  $[a_{N/2}, b_{N/2}]$ is length-minimizing, it must also be that $b_{N/2} < n - h_{\alpha/2}$, hence $n - h_{\alpha/2}\not\in [a_{N/2},b_{N/2}]$. Then by Lemma \ref{larger}, 
\begin{equation}
P_{N/2}(a_{N/2}) \geq P_{N/2}(n - h_{\alpha/2}) = P_{N/2}(h_{\alpha/2}),\label{inequality N/2}  
\end{equation} this last by symmetry. On the other hand, recall that $h_{\alpha/2}\le \lfloor n/2\rfloor$;  see \eqref{h<n/2}.  The mode~\eqref{argmax.x} in this case is $m= (n + 1)/2$. We have 
\begin{equation*}
a_{N/2}< h_{\alpha/2} \le \lfloor n/2\rfloor \leq 
\begin{cases} \lfloor m \rfloor - 1, & \mbox{if $n$ is odd ($m =  \lfloor m \rfloor $),}\\
\lfloor m \rfloor & \mbox{if $n$ is even ($m \neq  \lfloor m \rfloor $).} 
\end{cases}
\end{equation*} 
By this and Lemma~\ref{lem:hyp.uni.x} we have that  $P_{N/2}(h_{\alpha/2})\ge P_{N/2}(a_{N/2})$. Moreover, this inequality is strict by part~\ref{part:unimod.pos.PM} of Lemma~\ref{lem:hyp.uni.x} since the latter is positive, $a_{N/2}$ being an endpoint of an \AMO{} interval, hence  the former is positive too. The strict inequality contradicts \eqref{inequality N/2}.

If $b_{N/2} > n - h_{\alpha/2}$ then similar arguments apply.

For part~\ref{prt:a<h}, for $M\in[N/2-1]$, using Lemma~\ref{lemma: n - b > a} to show inclusion of the following intervals, we have
\begin{gather*}
P_{M} ( [a_M, n - a_M]) \geq P_{M} ( [a_M, b_M]) \geq 1 - \alpha,\qm{and}\\
P_{N -M} ( [a_M, n - a_M]) \geq P_{N - M} ( [n - b_M, n - a_M]) \geq 1 - \alpha.
\end{gather*}
The first  claim of part~\ref{prt:a<h} then follows from these inequalities and Lemma~\ref{lem:uni.intvl}, which says that $M'\mapsto P_{M'} ( [a_M, n - a_M])$ is unimodal, thus $P_{N/2} ( [a_M, n - a_M]) \geq 1 - \alpha$. This inequality is equivalent to
 $$\alpha\ge P_{N/2}(X < a_M ) + P_{N/2}(X > n - a_M ) =2P_{N/2}(X < a_M ),$$ by symmetry, and thus $a_M\le h_{\alpha/2}$ by definition of the latter, establishing the second claim.
\end{proof}

\begin{lemma}\label{lemma: n - b > a}
Let $\{[a_M, b_M] \mid M =  [\lfloor N/2 \rfloor] \}$ be \AMO{}. Then
\begin{equation}\label{lem.al.n-a>b}
a_M +b_M\le n\qmq{for all} M < N/2.
\end{equation}
\end{lemma}

\begin{proof} Suppose this fails, so that $b_{M^*} > n - a_{M^*}$  for some  $M^*\in\mM:=[\lfloor (N - 1)/2 \rfloor] $. By Lemma~\ref{lemma: opt sym}, \begin{equation}\label{eq:a largest 1}
\{[\wtilde{a}_M = n - b_{N - M}, \wtilde{b}_M  = n - a_{N - M}]:\; M \in\wtilde{\mM}:=N-\mM \} 
\end{equation}
is \AMO{} for $\wtilde{\mM}$. Then, since $$b_{M^*} \notin [n - b_{M^*},n -  a_{M^*}] = [\wtilde{a}_{N-M^*}, \wtilde{b}_{N-M^*}],$$
by  Lemma \ref{larger} we have that
 \begin{equation}\label{eq:a largest 2}
P_{N-M^*}( b_{M^*}) \leq P_{N - M^*} ( n - b_{M^*}).    
\end{equation} By similar arguments, since $n - b_{M^*} \notin [ a_{M^*}, b_{M^*}]$ we have that $P_{M^*}( b_{M^*}) \geq P_{M^*} ( n - b_{M^*})$. Next we will apply Lemma \ref{prob ratio} with $x_1=b_{M^*}$ and $x_2=n-b_{M^*}$. We have $$b_{M^*} \leq M^*< N/2 < N - M^* \leq  N - n + n - b_{M^*},$$ so that lemma tells us that
$$\frac{P_{N -M^*}( b_{M^*}) }{P_{N - M^*} ( n - b_{M^*})} >\frac{P_{M^*}( b_{M^*}) }{P_{M^*} ( n - b_{M^*})} \geq 1,$$
which contradicts \eqref{eq:a largest 2} and thus establishes \eqref{lem.al.n-a>b}.
\end{proof}

\section{Algorithm~\ref{alg:AO.ints}}\label{sec:alg2}

\begin{algorithm}[!htp]
\caption{Given $\alpha$, $n$, and $N$, calculate a set of level-$\alpha$ acceptance intervals~$\{[a^*_M,b^*_M]:\; M\in[N]\}$.}

\begin{algorithmic} 
\REQUIRE $N \in \mathbb{N}$, $n \leq N$ and $ 0 < \alpha < 1$

\FOR{$M = 0, ..., \lfloor N/2 \rfloor$ }
\STATE $x_{\min}= \max\{0,M+n-N\}$
\STATE $x_{\max}= \min\{n,N\}$
\STATE $C, D = \lfloor \frac{(n+1)(M+1)}{N+2}\rfloor$
\STATE $P = P_M(C)$
\STATE\textbf{if}  $C>x_{\min}$ \textbf{then} $PC= P_M(C-1)$ \textbf{else} $PC= 0$ \textbf{end if}
\STATE\textbf{if}  $D<x_{\max}$ \textbf{then} $PD = P_M(D+1)$ \textbf{else} $PD= 0$ \textbf{end if}
\WHILE{$P< 1 - \alpha$}
\IF{$PD>PC$}
\STATE $D= D+1$, $P= P+PD$
\STATE\textbf{if}  $D<x_{\max}$ \textbf{then} $PD = P_M(D+1)$ \textbf{else} $PD= 0$ \textbf{end if}
\ELSE 
\STATE $C= C-1$, $P= P+PC$
\STATE\textbf{if}  $C>x_{\min}$ \textbf{then} $PC= P_M(C-1)$ \textbf{else} $PC= 0$ \textbf{end if}
\ENDIF
\ENDWHILE
\STATE $a_M = C$, $b_M = D$
\ENDFOR
\STATE $b^*_0  = b_0$, $a^*_0  = a_0$
\FOR {$M = 1, \ldots, \lfloor N/2 \rfloor$ }
\IF {$a_{M} < a^*_{M - 1}$}
\STATE $a^*_M  =  a^*_{M - 1}$, $b^*_M  = b_M + a^*_{M - 1} - a_{M}$
\ELSE 
\STATE $b^*_M  = b_M$, $a^*_M  = a_M$
\ENDIF 
\ENDFOR
\FOR {$M =  \lfloor N/2 \rfloor$ - 1,\ldots, 0 }
\IF {$b^*_{M} > b^*_{M + 1}$}
\STATE $a^*_M  = a^*_M + b^*_{M + 1} - b^*_{M}$, $b^*_M  =  b^*_{M + 1}$
\ENDIF 
\STATE $a^*_{N - M} = n - b^*_M$, $b^*_{N - M} = n - a^*_M$
\ENDFOR
\IF{$N$ is even}
\STATE $a^*_{N/2} = \max\left\{a_{N/2}, n-b_{N/2}\right\}$, $b^*_{N/2} = n - a^*_{N/2}$
\ELSE
\STATE $a^*_{ \lfloor N/2 \rfloor + 1} = n - b^*_{ \lfloor N/2 \rfloor}$, $b^*_{ \lfloor N/2 \rfloor +1} = n - a^*_{ \lfloor N/2 \rfloor}$

\ENDIF

\RETURN $\{[a^*_M,b^*_M]\}^{N}_{M = 0}$ 
\end{algorithmic}
\label{alg:AO.ints}
\end{algorithm}

\section{Auxiliary results for Section~\ref{sec:length.CI}}\label{sec:len.opt.aux}

\begin{lemma}\label{lem:A*<ACM} In the setting of Theorem~\ref{thm:len.opt.set}, for any $\mC\in\fC_S$ and  $M\ne N/2$, 
\begin{equation*}
|\mA^*(M)|\le|\mA_{\mC}(M)|.
\end{equation*}
\end{lemma}

\begin{proof} Fix $M\in[N]$, $M\ne N/2$. By an argument similar to the proof of the converse part of Lemma~\ref{larger}, there is an interval $[a_M, b_M]$ such that $b_M - a_M + 1 = |\mA_{\mC}(M)|$  and $P_M( [a_M, b_M]) \geq P_M(\mA_{\mC}(M))$. Since $M\ne N/2$, $\mA^*(M)=[a_M^*,b_M^*]$ is size-optimal by Theorem~\ref{thm:a*b*} so  
\begin{equation}\label{A*<A}
|\mA^*(M)|=b_M^* - a_M^*+1\le b_M - a_M +1= |\mA_{\mC}(M)|.
\end{equation} 
\end{proof}

\begin{lemma}\label{lemma: interval 2 }
For even $N$, assume $A\subseteq [n]$ is nonempty and such that $x \in A \Rightarrow n - x \in A$. Then there exists $c\in[n]$ such that 
\begin{equation}\label{lem.Pcn-c>A}
P_{N/2}( [c, n - c]) \geq P_{N/2}( A)
\end{equation}
 and
 \begin{equation*}
n - 2c + 1= \begin{cases}
|A|,&\mbox{if $n$   or $|A|$ is odd,}\\
|A|+1,&\mbox{if $n$   and $|A|$ are even.}
\end{cases}
\end{equation*}
\end{lemma}

\begin{proof} 
If $n$ is odd then $x \neq n - x$ for all $x \in[n]$, implying that $|A|$ is even. Let $c = (n - |A| + 1)/2$, an integer.  If $|A|=0$ then there is nothing to prove for \eqref{lem.Pcn-c>A}, so assume $|A|\ge 2$, whence $c\le (n-1)/2=:m$, which by \eqref{argmax.x} is the mode of the $\mbox{Hyper}(N/2, n, N)$ density.  Thus $m\in[c,n-c]$ and by \eqref{X.n-X.sym} this density takes the same value at the endpoints~$c$ and $n-c$.   Combining these facts implies that  $P_{N/2}(x_1) \geq P_{N/2}(x_2)$ for any $x_1 \in [c, n - c]$ and $x_2 \notin [c, n - c]$. Then $P_{N/2}( [c, n - c]) \geq P_{N/2}( A)$ now follows from this and the fact that these two sets have the same number of points, $n - 2c + 1 = |A|$. 
 
 If $n$ is even and $|A|$ is odd, then $c = (n - |A| + 1)/2$ is still an integer. If $|A|=1$ then $[c,n-c]=\{n/2\}$, the point maximizing  $P_{N/2}(\cdot)$ by Lemma~\ref{lem:hyp.basics}, hence  \eqref{lem.Pcn-c>A} holds.  Otherwise, $|A|\ge 3$ and $c<m$ so the argument in the previous paragraph applies.
 
If $n$ and $|A|$ are both even, let $c = (n - |A|)/2$, an integer,  and $B=[c, n - c - 1]$.  By unimodality and symmetry of $P_{N/2}(\cdot)$ about $n/2$ we have that $$\min_{x\in B}P_{N/2}(x)=P_{N/2}(c)= P_{N/2}(n-c) = \max_{x\not\in B} P_{N/2}(x).$$ It follows from this and $|B|=n-2c=|A|$ that $P_{N/2}(B) \geq P_{N/2}( A)$, thus $P_{N/2}( [c, n - c]) \geq P_{N/2}(B) \geq P_{N/2}( A)$.
\end{proof}

\begin{lemma}\label{lem:Aodd.C*opt}
In the setting of Theorem~\ref{thm:len.opt.set}, suppose $n$ and $N$ are even. Then, for $\mC\in\fC_S$,
\begin{equation*}
|\mC^*|\le \begin{cases}
|\mC|,&\mbox{if $|\mA_{\mC}(N/2)|$ is odd,}\\
|\mC|+1,&\mbox{if $|\mA_{\mC}(N/2)|$ is even.}\\
\end{cases}
\end{equation*}
\end{lemma}

\begin{proof} We have $|\mA^*(M)|\le|\mA_{\mC}(M)|$ for all $M\ne N/2$ by Lemma~\ref{lem:A*<ACM}. $\mA_{\mC}(N/2)$ is symmetrical, so by Lemma~\ref{lemma: interval 2 } there is  an interval~$[a, n-a]$ such that $P_{N/2}( [a, n - a]) \geq P_{N/2}( \mA_{\mC}(N/2))$ and 
\begin{equation}\label{n-2a+1<}
n - 2a + 1 \le \begin{cases}
|\mA_{\mC}(N/2)|,&\mbox{if $|\mA_{\mC}(N/2)|$ is odd,}\\
|\mA_{\mC}(N/2)|+1,&\mbox{if $|\mA_{\mC}(N/2)|$ is even.}
\end{cases}
\end{equation}
Since $\mA^*(N/2)=[a_{N/2}^*, b_{N/2}^*] = [a_{N/2}^*, n  - a_{N/2}^*]$ is the shortest symmetrical acceptance interval for $M = N/2$, we have  
\begin{equation*}
|\mA^*(N/2)|=b_{N/2}^* - a_{N/2}^*+1\le n-2a +1,
\end{equation*} which is thus $\le$ the right-hand-side of \eqref{n-2a+1<}.   This, with the above inequality for the $M\ne N/2$ cases, gives the desired result after summing in an argument like \eqref{C>C*.set}.
\end{proof}

%\begin{lemma} \label{lemma: interval max} For any $A\subseteq[n]$, there exists an interval $[a, b]\subseteq[n]$ such that $b - a + 1 = |A|$  and $P_M( [a, b]) \geq P_M( A)$. 
%\end{lemma}
%
%\begin{proof}
%Because $P_M( x)$ is unimodal with respect to $x$, we can find an interval $[a, b]$ satisfying $b - a + 1 = |A|$ and $P_M(x_1) \geq P_M(x_2)$ for any $x_1 \in [a, b]$ and $x_2 \notin [a, b]$. 
%Let $B = \{x = 0, \ldots, n \mid x \in [a, b]/A\} $. Then, $|B| = |A\setminus [a,b]|$, since $b - a + 1 = |A|$. Therefore,
%\begin{align*}
%P_M( A) & = P_M( A \cap [a,b]) + \sum_{x \in A\setminus [a,b]}P_M(x)\\
%& \leq P_M( A \cap [a,b]) + \sum_{x \in B}P_M(x) = P_M( [a, b]).
%\end{align*}
%\end{proof}

\end{document}